\documentclass[11pt]{article}
\usepackage[english]{babel}
\usepackage[utf8]{inputenc}
\usepackage{lmodern}
\usepackage[T1]{fontenc}
\usepackage{amsmath,amssymb,lmodern,amsthm}
\usepackage{marginnote}
\usepackage{mathtools}
\usepackage{graphicx}
\usepackage{mathrsfs}
\usepackage{bm}
\usepackage{url}
\usepackage{breakurl}
\usepackage{array}
\usepackage{tensor}
\usepackage{pifont}
\usepackage[toc,page]{appendix}
\usepackage[left=2.5cm,right=2.5cm,top=3cm,bottom=3cm]{geometry}
\usepackage{tikz}
\usetikzlibrary{arrows,positioning}
\usepackage{enumitem}
\usepackage{empheq}

\usepackage{cases}
\usepackage{csquotes}
\usepackage{accents}
\usepackage{setspace}

\theoremstyle{plain}
\usepackage{xargs}
\usepackage{xcolor}
\usepackage{subcaption}
\usepackage[justification=raggedright]{caption}
\usepackage{hyperref}
\usepackage{cleveref}
\usepackage{tocloft}

\hypersetup{
	colorlinks=true,
	urlcolor=blue,
	linkcolor={black},
	citecolor={green!40!black}
}
\usepackage[affil-it]{authblk}
\usepackage{thmtools}
\usepackage[backend=bibtex,style=numeric-comp,maxnames=4,sorting=none,url=false,hyperref=true,date=year,eprint=true,doi=true]{biblatex}

\newlist{properties}{enumerate}{10}
\setlist[properties]{label*=\roman*)}
\crefname{propertiesi}{\text{property}}{\text{properties}}
\Crefname{propertiesi}{\text{Property}}{\text{Properties}}
\newlist{points}{enumerate}{10}
\setlist[points]{label*=\arabic*)}
\crefname{pointsi}{\text{point}}{\text{points}}
\Crefname{pointsi}{\text{Point}}{\text{Points}}

\newcommand{\ubar}[1]{\underaccent{\bar}{#1}}

\numberwithin{equation}{section}
\newtheorem{thm}{Theorem}[]
\theoremstyle{definition}

\newtheorem{corollary}{Corollary}[section]
\newtheorem{deff}{Definition}[section]
\newtheorem*{crit*}{Criterion}

\renewcommand{\iff}{\Longleftrightarrow}

\newcommand{\lied}{\pounds}
\newcommand{\defeq}{\vcentcolon=}

\newcommand{\scri}{\mathscr{J}}

\newcommand{\pt}[2]{\tensor{\hat{#1}}{#2}}

\newcommand{\ctr}[3]{\tensor[{#1}]{#2}{#3}}

\newcommand{\ct}[2]{\tensor{{#1}}{#2}}

\newcommand{\ctl}[2]{\tensor[^{{\,}^{\scriptscriptstyle\!\ell\!}\!}]{{#1}}{#2}}
\newcommand{\ctlcn}[2]{\tensor[^{{\,}^{\scriptscriptstyle\!\ell\!}\!}]{{\ubar{#1}}}{#2}}
\newcommand{\ctn}[2]{\tensor[^{{\,}^{\scriptscriptstyle\!N\!}\!}]{{#1}}{#2}}
\newcommand{\ctncn}[2]{\tensor[^{{\,}^{\scriptscriptstyle\!N\!}\!}]{{\ubar{#1}}}{#2}}
\newcommand{\ctkcn}[2]{\tensor[^{{\,}^{\scriptscriptstyle\!k\!}\!}]{{\ubar{#1}}}{#2}}

\newcommand{\ctt}[2]{\tensor{\grave{#1}}{#2}}
\newcommand{\cttcn}[2]{\tensor{\grave{\ubar{#1}}}{#2}}
\newcommand{\cts}[2]{\tensor{\overline{#1}}{#2}}

\newcommand{\cscn}[1]{\ubar{{#1}}}
\newcommand{\ctcn}[2]{\tensor{\ubar{{#1}}}{#2}}
\newcommand{\ctcnr}[3]{\tensor[{#1}]{\ubar{#2}}{#3}}

\newcommand{\cdcn}[1]{\tensor{\ubar{\D}}{#1}}

\newcommand{\ctcng}[2]{\tensor{\tilde{\ubar{#1}}}{#2}}
\newcommand{\ctg}[2]{\tensor{\tilde{#1}}{#2}}

\newcommand{\cta}[2]{\tensor{{#1}}{^\prime#2}}
\newcommand{\ctsa}[2]{\tensor{{\overline{#1}}}{^\prime#2}}
\newcommand{\ctcna}[2]{\tensor{{\ubar{#1}}}{^\prime#2}}

\newcommand{\cs}[1]{#1}

\newcommand{\ps}[1]{\hat{#1}}

\newcommand{\df}[1]{\text{d}#1}

\newcommand{\prn}[1]{\left(#1\right)}

\newcommand{\brkt}[1]{\left[#1\right]}

\newcommand{\cbrkt}[1]{\left\lbrace#1\right\rbrace}

\newcommand{\evalat}[1]{\Big|_{#1}}
\newcommand{\eqs}{\stackrel{\scri}{=}}
\newcommand{\eqn}{\stackrel{\N}{=}}

\newcommand{\cd}[1]{\tensor{\nabla}{#1}}

\newcommand{\cds}[1]{\tensor{\overline{\nabla}}{#1}}
\newcommand{\cdc}[1]{\tensor{\D}{#1}}
\newcommand{\pd}[1]{\tensor{\hat{\nabla}}{#1}}

\newcommand{\dpart}[1]{\partial_{#1}}

\newcommand{\spacef}{\ }

\newcommand{\D}{\mathcal{D}}

\newcommand{\Sc}{\mathcal{S}}

\newcommand{\N}{\mathcal{N}}

\newcommand{\ms}[1]{\cts{g}{#1}}

\newcommand{\mc}[1]{\ct{q}{#1}}

\newcommand{\mcn}[1]{\ct{\ubar{q}}{#1}}

\newcolumntype{M}[1]{>{\centering\arraybackslash}m{#1}}
\newcolumntype{N}{@{}m{0pt}@{}}

\newcounter{marginnotecount}[section]

	\title{\Large \textbf{News tensor on null hypersurfaces}}
	\author[]{Francisco Fernández-Álvarez\thanks{francisco.fernandez@ehu.eus}\ }
	\affil[]{Departamento de Física\\ Universidad del País Vasco UPV/EHU\\ Apartado 644, 48080 Bilbao, Spain}
	\date{\today{}}
\begin{document}
	
	\maketitle

	\begin{abstract}
	 A geometric definition of news tensor on null hypersurfaces  in four space-time dimensions is presented. When the conformal Einstein field equations hold, this news tensor yields the correct expression for the radiative components of the rescaled Weyl tensor at infinity  with vanishing cosmological constant in arbitrary conformal gauge. Also, a generalised transport equation for the Geroch tensor is derived. Important differences between null hypersurfaces in the bulk of the space-time and null infinity with vanishing cosmological constant are reviewed, and their impact on the role of the news is discussed.
	\end{abstract}

	\tableofcontents
\section{Introduction}
The theory of gravitational radiation at infinity has been developed for decades now. In full general relativity, great achievements were made in the 50s and 60s \cite{Pirani57,Bel1958,Trautman58,Bondi1962,Sachs1962,Newman62,Winicour1966}, and also in the following decades \cite{Geroch1977,Ashtekar1981c,Geroch81}. For the case of a positive cosmological constant $ \Lambda $, important open problems remain to be solved \cite{Penrose2011} and in the last 10 years advances have been made \cite{Ashtekar2014,Fernandez-AlvarezSenovilla2020b,} ---see also \cite{Szabados2019,Senovilla2022} for a review and references therein. When $ \Lambda=0 $, an object called `news' is central in the analysis of asymptotic gravitational radiation: it vanishes only in the absence of gravitational waves and is a key piece in the Bondi-Trautman energy loss formula and in the expressions of other charges and balance laws. This `news' object has different faces: in some works it is a function \cite{Bondi1962}, sometimes identified with the components of the asymptotic shear \cite{Madler2016} or the  complex scalar shear in Newman-Penrose formalism \cite{Penrose62,Newman2009}; in others, it is a gauge invariant covariant quantity \cite{Geroch1977,Ashtekar81,Fernandez-AlvarezSenovilla2022a}. In the analysis of gravitational radiation at infinity with $ \Lambda=0 $, the very special structure of the conformal boundary $ \scri $ ---in Penrose's geometric description of infinity \cite{Penrose65}--- plays a fundamental role. In the last years there has been interest in understanding the properties of null hypersurfaces \emph{in the bulk} by using ingredients similar to the ones present at infinity ---and vice versa---, such as supertranslations on horizons \cite{Hawking2016,Sousa2018}, comparisons with weakly isolated horizons \cite{Ashtekar2024a,Ashtekar2024b,Ashtekar2024c}, and also other studies like Carollian-geometry methods \cite{Riello2024} ---see also references therein.\\

The aim of this paper is twofold. Firstly, a tensorial definition of news on null hypersurfaces is put forward. For the case of $ \scri $, this definition gives the correct gauge invariant tensor and generalises  to arbitrary conformal frames at infinity the equations for the rescaled Weyl tensor in terms of news. Also, its expression by means of the different shears is given explicitly. Having expressions which are conformal-gauge independent is important, as particular choices of gauges ---such as the so called divergence-free gauge--- can be inadequate to perform numerical relativity \cite{Frauendiener2004}. In addition to that, a generalised transport equation for the so called Geroch tensor is found at $ \scri $. Secondly, a study of how this news tensor is related to the radiative components of the Weyl tensor on a null hypersurface $ \N $ in the bulk of the space-time is presented. This is a natural question to ask  if one aims to make an analogy with null infinity $ \scri $ with $ \Lambda=0 $. The result is that, in general, the news cannot determine the gravitational radiation crossing (or `outgoing') $ \N $, unless one makes further assumptions on the geometry and/or on the tangential (`incoming') radiation to $ \N $. Two particular cases are used to illustrate this: non-expanding horizons (NEH) and a special case of a null marginally trapped tube (MTT).
	\subsection{Conventions}\label{sec:conventions}
	In the rest of the paper 4 space-time dimensions are considered, with Greek letters $ \alpha, \beta... $ used for 4-dimensional abstract indices, keeping Latin ones $ a, b ... $ and $ A, B ... $ for three- and two-dimensional manifolds, respectively. The space-time metric has signature $ \prn{-,+,+,+} $ and the space-time covariant derivative is denoted $ \cd{_{\alpha}} $, defining the Riemann tensor such that $ \ct{R}{_{\alpha\beta\gamma}^{\delta}}\ct{v}{_{\delta}}=\prn{\cd{_{\alpha}}\cd{_{\beta}}-\cd{_{\beta}}\cd{_{\alpha}}}\ct{v}{_{\gamma}} $. \Cref{sec:news-null-hyp} does not make use of Einstein field equations (EFE) but in \cref{sec:conf-space} the conformal EFE (CEFE) are assumed. Also, the term `crossing radiation' on a foliated null hypersurface $ \N $ will be used to denote the part of the (rescaled) Weyl tensor containing $ \phi_{3} $ and $ \phi_{4} $ in Newman-Penrose notation with respect to a tetrad containing the null vector tangent to the generators of  $ \N $. In addition to that, `tangential radiation' will refer to the  $ \phi_{0} $ and $ \phi_{1} $ in that tetrad, choosing the other null direction  orthogonal to a given cross-section of $ \N $.
	\subsubsection*{Conformal completion}
			The physical space-times $ \prn{\ps{M},\pt{g}{_{\alpha\beta}}} $ that are considered in \cref{sec:conf-space} admit a conformal completion (unphysical space-time) $ \prn{\cs{M},\ct{g}{_{\alpha\beta}}} $ {\em \`a la} Penrose with boundary $ \scri $:
				\begin{properties}
					\item 	There exists an embedding $ \phi: \ps{M} \rightarrow \cs{M} $ such that $ \phi(\ps{M})=\cs{M}\setminus\scri $, and the physical metric is related to the conformal one as
					\begin{equation}
					\ct{g}{_{\alpha\beta}}= \Omega^2\pt{g}{_{\alpha\beta}},
					\end{equation} 
					where the pullback of the conformal metric to the physical space-time, $ \prn{\phi^*g}_{\alpha\beta} $, is referred by  $ \ct{g}{_{\alpha\beta}} $.
					\item $\Omega>0$ in $M\setminus\scri$ , $\Omega=0$ on $\scri$ and $\ct{N}{_\alpha}\defeq\cd{_\alpha}\Omega$ (the normal to $ \scri $) is non-vanishing there.\label{it:OmegaAtScri}
					\item $ \pt{g}{_{\alpha\beta}} $ is a solution of Einstein field equations with arbitrary $ \Lambda $ (unless otherwise specified).
					\item The energy-momentum tensor, $ \pt{T}{_{\alpha\beta}} $, vanishes at $ \scri $ and $ \ct{T}{_{\alpha\beta}}\defeq\Omega^{-1}\pt{T}{_{\alpha\beta}} $ is smooth there.\label{it:energytensorassumption}
				\end{properties}
			The remaining gauge freedom is to perform conformal rescalings of the conformal factor
				\begin{equation}
					\Omega \rightarrow \omega\Omega
				\end{equation}
			with $ \omega>0 $. For further details on this kind of completions see for example \cite{Kroon2016}.	
\section{News tensor on foliated null hypersurfaces}\label{sec:news-null-hyp}
	For this section, no field equations are assumed. The goal is to present a definition of news tensor for general null hypersurfaces equipped with a foliation by closed surfaces. This last requirement is natural because the news tensor is, by definition, an object associated to two-dimensional closed surfaces \cite{Geroch1977,Fernandez-AlvarezSenovilla2022a} and the concept of gravitational radiation is tightly related to quasi-local methods involving integrals over two-spheres \cite{Penrose1986,Szabados2004}. 
	\subsection{Set up of the null hypersurface equipped with a foliation}
	In this first part, the tools needed on null hypersurfaces are presented, introducing a foliation and notation that appears in the rest of the paper, and also developing the set up and obtaining key formulae according to the present interests and conventions ---for the study of general hypersurfaces, see \cite{Mars1993b}. Let $ \prn{M,\ct{g}{_{\alpha\beta}}} $  be a 4-dimensional space-time\footnote{In \cref{sec:conf-space},  the conformal compactification of a physical space-time is considered. So far, in this section, no field equations at all are assumed.} and let $ \prn{\N,\ms{_{ab}}} $ be a null hypersurface $ \N $ with degenerate first fundamental form $ \ms{_{ab}} $. Denote by $ \ct{k}{^a} $ the degeneration vector, $ \ct{k}{^b}\ms{_{ab}}=0 $ everywhere in $ \N $ ---of course, $ \ct{k}{^{a}} $ can always be rescaled. Also, assume that $ \N $ admits a foliation by two-dimensional leaves $ \Sc_{C} $ labelled by constant values $ C $ of a function $ F $ on $ \N $ satisfying
		\begin{equation}\label{eq:foliation-F}
			\dot{F}\defeq \ct{k}{^a}\dpart{a}F\neq 0
		\end{equation}
	and such that
		\begin{align}
			\Sc_{C_{1}}\cap\Sc_{C_{2}}&=\emptyset\ ,\label{eq:foliation-cap}\\
			\bigcup_{C}\Sc_{C}&=\N\ \label{eq:foliation-cup}.
		\end{align}
	Define the one-form
		\begin{equation}\label{eq:ell}
			\ct{\ell}{_{a}}\defeq -\frac{1}{\dot{F}}\dpart{a}F\ ,\quad\ct{k}{^p}\ct{\ell}{_{p}}=-1\ .
		\end{equation}
	Let $ \mathfrak{X}_{\N} $, $ \Lambda_{\N} $, $ \mathfrak{X}_{\Sc} $ and $ \Lambda_{\Sc} $ be the set of tangent vector fields and forms on $ \N $ and on the leaves of the foliation, respectively. Then
		\begin{align}
			\ct{v}{^a}\ct{\ell}{_{a}}=0 &\iff \ct{v}{^a}\in \mathfrak{X}_{\Sc}\ ,\\	\ct{w}{_a}\ct{k}{^{a}}=0 &\iff \ct{w}{_a}\in \Lambda_{\Sc}\ .
		\end{align}
	Let $ \cbrkt{\ctcn{E}{^a_{A}}} $ and $ \cbrkt{\ctcn{W}{_{a}^A}} $ be sets of 2 linearly-independent vector fields and forms, constituting a pair of orthonormal bases for $ \mathfrak{X}_{\Sc} $ and $ \Lambda_{\Sc} $. Observe that $ \cbrkt{-\ct{\ell}{_{a}},\ctcn{W}{_{a}^A}} $ and $ \cbrkt{\ct{k}{^a},\ctcn{E}{^a_{A}}} $ are bases for $ \Lambda_{\N} $ and $ \mathfrak{X}_{\N} $, respectively. Then, the projector to the leaves can be defined as
		\begin{equation}\label{eq:projector-leaves}
			\ctcn{P}{^a_{b}}\defeq\ctcn{E}{^a_{D}}\ctcn{W}{_{b}^D}=\delta^a_{b}+\ct{k}{^a}\ct{\ell}{_{b}}\ ,\quad\ctcn{P}{^a_{b}}\ct{\ell}{_{a}}=0=\ctcn{P}{^a_{b}}\ct{k}{^b}\ ,\quad\ctcn{P}{^a_{a}}=2\ .
		\end{equation}
	The notation $ \ctcn{v}{_{a}}\defeq\ctcn{P}{^b_{a}}\ct{v}{_{b}}$ for $ \ct{v}{_{a}}\in \Lambda_{\N} $ will appear frequently. More generally, an underbar will denote that a tensor $ \ctcn{A}{_{abc...}^{efg...}} $ is tangent to the leaves (orthogonal to $ \ct{\ell}{_{a}} $ and $ \ct{k}{^b} $ in all its indices). Observe that $ \ctcn{E}{^a_{A}} $ and $ \ctcn{W}{_{a}^A} $ are vector fields on $ \N $, i.e., they are defined for every leaf of the foliation and depend on $ F $. Thus, in general, given $ \ctcn{A}{_{a }^{b}} $ one can write
		\begin{equation}
			 \ctcn{A}{_{a}^{b}} =\ctcn{E}{^b_{B}}\ctcn{W}{_{a}^A}\ctcn{A}{_{A}^B}\ ,
		\end{equation}
	but the quantity $ \ctcn{A}{_{A}^B} $ depends on $ F $. For $ F=C $, $ \ctcn{A}{_{A}^B} $ is a tensor field on the leaf $ \Sc_{C} $. It is easy to show that
		\begin{equation}\label{eq:lie-ell-k}
			\lied_{\vec{k}}\ct{\ell}{_{b}}=-\ctcn{P}{^a_{b}}\dpart{a}\ln \dot{F}\ .
		\end{equation}
	Also, introduce two-dimensional coordinates $ \zeta^A$, and choose $ \ctcn{W}{_{a}^ A}=\prn{\df{\zeta}^A}_{a} $. This is the choice used from now on. With it, one has
		\begin{align}
			\lied_{\vec{k}}\ctcn{E}{^a_{A}}&=-\ct{k}{^a}\ctcn{E}{^m_{A}}\dpart{m}\ln \dot{F}\ ,\label{eq:liedE}\\
			\lied_{\vec{k}}\ctcn{W}{_{a}^A}&=0\ .
		\end{align}
	 Note that using \cref{eq:liedE}
			\begin{equation}
				\ctcn{E}{^{a}_{A}}\ctcn{E}{^{b}_{B}}\lied_{\vec{N}}\ctcn{A}{_{ab}}=\lied_{\vec{N}}\ctcn{A}{_{AB}}\ .
			\end{equation}
	In addition to this, define an object $ \ms{^{ab}} $ as the \emph{unique} tensor field satisfying
		\begin{equation}
			\ms{_{da}}\ms{^{dc}}\ms{_{cb}}=\ms{_{ab}}\ ,\quad\ct{\ell}{_{a}}\ms{^{ab}}=0\ .S
		\end{equation}
	Observe that the second of these equations fixes the ambiguity $ \ms{^{dc}}+\ct{v}{^{(d}}\ct{k}{^{c)}}$ that generates a new tensor fulfilling the first condition. Then, $ \ctcn{g}{_{ab}} $ and $ \ctcn{g}{^{ab}} $ constitute essentially a two dimensional metric and its inverse --- $ \ms{_{ab}}\ms{^{ab}}=2 $ ---, and it is possible to rise and lower indices of tensors orthogonal both to $ \ct{\ell}{_{a}} $ and $ \ct{k}{^a} $. Nevertheless, special care has to be put when contracting an arbitrary $ \ct{v}{_{a}}\in\Lambda_{\N} $, because in general
		\begin{equation}
			\ms{_{da}}\ms{^{dc}}\ct{v}{_{c}}=\ctcn{v}{_{a}}\neq\ct{v}{_{a}}\ ,
		\end{equation}
	indeed,
		\begin{equation}
			\ms{^{ad}}\ms{_{bd}}=\ctcn{P}{^a_{b}}\ .
		\end{equation}
	It is also convenient to introduce
		\begin{align}
			\mcn{_{AB}}&\defeq\ctcn{E}{^a_{A}}\ctcn{E}{^b_{B}}\ms{_{ab}}\ ,\\
			\mcn{^{AB}}&\defeq\ctcn{W}{_a^{A}}\ctcn{W}{_b^{B}}\ms{^{ab}}\ .
		\end{align}
	One can think of $ \mcn{_{AB}} $ as a one-parameter family of two-dimensional metrics. Then, when $ F $ is restricted to a constant $ F=C $, $ \mcn{_{AB}} $ is the Riemannian metric on the two-dimensional leaf $ \Sc_{C} $. \\
	
	Next, from the viewpoint of the space-time $ \prn{M,\ct{g}{_{\alpha\beta}}} $, let $ \ct{k}{_{\alpha}} $ be normal to $ \N $. Consider a set of 3 linearly-independent, orthonormal vector fields $ \cbrkt{\ct{e}{^\alpha_{a}}} $ orthogonal to $ \ct{k}{_{\alpha}} $ defined on $ \N $. The degenerate first fundamental form is related to the space-time metric $ \ct{g}{_{\alpha\beta}} $ via
		\begin{equation}
			\ms{_{ab}}=\ct{e}{^\alpha_{a}}\ct{e}{^\beta_{b}}\ct{g}{_{\alpha\beta}}.
		\end{equation}
	Also, any vector field defined at least on $ \N $ obeying $ \ct{v}{^\alpha}\ct{k}{_{\alpha}}=0 $ is determined there by a vector field $ \ct{v}{^a}\in\mathfrak{X}_{\N} $,
		\begin{equation}
			\ct{v}{^\alpha}= \ct{v}{^a}\ct{e}{^\alpha_{a}}\ , \text{ where}\ \quad\ct{v}{^\alpha}\ct{k}{_{\alpha}}=0.
		\end{equation}
	Notice that because $ \N $	is a null hypersurface, one has
		\begin{equation}
			\ct{k}{^\alpha}\ct{k}{_{\alpha}}=0\ ,
		\end{equation}
	and it is possible to write
		\begin{equation}
			\ct{k}{^\alpha}=\ct{e}{^\alpha_{a}}\ct{k}{^a}\ .
		\end{equation}
	To decompose the tangent plane at $ \N $, one introduces a \emph{rigging vector} $ \ct{\ell}{^\alpha} $. Here, since one is considering the lightlike hypersurface $ \N $ to be equipped with a foliation, the most convenient choice in this work is to adapt the rigging to be the unique lightlike vector field (other than $ \ct{k}{^\alpha} $) that is orthogonal to the leaves, i.e.,
		\begin{equation}
			\ct{e}{^\alpha_{a}}\ct{\ell}{_{\alpha}}=\ct{\ell}{_{a}}\ ,\quad \ct{\ell}{_{\alpha}}\ct{\ell}{^\alpha}=0\ ,\quad\ct{\ell}{^\alpha}\ct{k}{_{\alpha}}=-1\ ,\quad\ct{\ell}{_{\alpha}}\ctcn{E}{^\alpha
			_{A}}=0\ ,
		\end{equation}
	where $ \ctcn{E}{^\alpha_{A}}=\ct{e}{^\alpha_{a}}\ctcn{E}{^{a}_{A}} $ and $ \ct{\ell}{_{a}} $ was defined in \cref{eq:ell}. Then, write a basis and its dual at points of $ M $ at $ \N $ as $ \cbrkt{\ct{\ell}{^\alpha},\ct{e}{^\alpha_{a}}} $ and $ \cbrkt{-\ct{k}{_{\alpha}},\ct{\omega}{_{\alpha}^{a}}} $, with $ \cbrkt{\ct{\omega}{_{\alpha}^a}}=\cbrkt{-\ct{\ell}{_{\alpha}},\ctcn{W}{_{\alpha}^A}} $  and $ \cbrkt{\ct{e}{^\alpha_{a}}}=\cbrkt{\ct{k}{^\alpha},\ct{e}{^\alpha_{a}}} $, satisfying
		\begin{equation}
		\ctcn{W}{_{\alpha}^A}\ctcn{E}{^\alpha_{B}}=\delta^A_{B}\ ,\quad \ct{e}{^\alpha_{a}}\ct{\omega}{_{\alpha}^b}=\delta^b_{a}\ ,\quad \ct{k}{_{\alpha}}\ct{e}{^\alpha_{a}}=0=\ct{\omega}{_{\alpha}^a}\ct{\ell}{^\alpha}\ .
		\end{equation}
	The projector to $ \N $ can be defined as
		\begin{equation}
			\ct{P}{^\alpha_{\beta}}\defeq \ct{e}{^\alpha_{m}}\ct{\omega}{_{\beta}^m}=\delta^\alpha_{\beta}+\ct{k}{_{\beta}}\ct{\ell}{^{\alpha}}\ .
		\end{equation}
	Observe also that $ \ctcn{W}{_{\alpha}^A}=\ct{\omega}{_{\alpha}^m}\ctcn{W}{_{m}^A}$ and one can write the space-time version of the projector to the leaves of the foliation
		\begin{equation}
			\ctcn{P}{^\alpha_{\beta}}=\ctcn{E}{^\alpha_{M}}\ctcn{W}{_{\beta}^M}=\delta^\alpha_{\beta}+\ct{k}{^{\alpha}}\ct{\ell}{_{\beta}}+\ct{\ell}{^\alpha}\ct{k}{_{\beta}}\ .
		\end{equation}
	The second fundamental form of $ \N $ reads
		\begin{equation}
			\ct{K}{_{ab}}=\ct{e}{^\alpha_{a}}\ct{e}{^\beta_{b}}\cd{_\alpha}\ct{k}{_{\beta}}=\frac{1}{2}\lied_{\vec{k}}\ms{_{ab}}\ ,
		\end{equation}
	where $ \cd{_{\alpha}} $ is the space-time Levi-Civita connection. Also, the following three-dimensional volume form will be used,
		\begin{equation}
			-\ct{k}{_\alpha}\ct{\epsilon}{_{abc}}=\ct{\eta}{_{\alpha\mu\nu\sigma}}\ct{e}{^\mu_{a}}\ct{e}{^\nu_{b}}\ct{e}{^\sigma_{c}}\ ,
		\end{equation}
	where $ \ct{\eta}{_{\alpha\beta\gamma\delta}} $ is the space-time volume form. The contravariant version is determined by $ \ct{\epsilon}{^{abc}}\ct{\epsilon}{_{abc}}=6 $. For the leaves of the foliation, one can introduce
		\begin{align}
		-\ct{\ell}{_{a}}\ctcn{\epsilon}{_{AB}}=\ct{\epsilon}{_{amn}}\ctcn{E}{^m_{A}}\ctcn{E}{^n_{B}}\ ,\\
		\ct{k}{^a}\ctcn{\epsilon}{^{AB}}=\ct{\epsilon}{^{amn}}\ctcn{W}{_{m}^A}\ctcn{W}{_{n}^B}\ ,
		\end{align}
	with the orientation fixed to $\ctcn{\epsilon}{_{23}}= \ct{\epsilon}{_{123}}=\ct{\eta}{_{0123}}=1 $.  \\
	
	The next step is to introduce a connection on $ \N $, and for that it is convenient to define some `kinematic' quantities,
		\begin{align}
		\ct{H}{_{ab}}\defeq \ct{e}{^\alpha_{a}}\ct{e}{^\beta_{b}}\cd{_{\alpha}}\ct{\ell}{_{\beta}}\ ,\quad \psi^a_{b}\defeq\ct{\omega}{^a_{\mu}}\ct{e}{^\nu_{b}}\cd{_{\nu}}\ct{\ell}{^\mu}\ ,\quad \ct{\varphi}{_{a}}\defeq \ct{k}{^m}\ct{H}{_{am}}\ ,
		\end{align}
	together with the two-dimensional divergence tensor, the expansion and the  shear of $ \ct{k}{_{\alpha}} $,
		\begin{align}
		 	\ctcn{\kappa}{_{AB}}&\defeq\ctcn{E}{^\alpha_{A}}\ctcn{E}{^\beta_{B}} \cd{_{\alpha}}\ct{k}{_{\beta}}\ ,\\
		 	\ctcn{\kappa}{}&\defeq \mcn{^{AB}}\ctcn{\kappa}{_{AB}}\ ,\\
		 	\ctcn{\nu}{_{AB}}&\defeq\ctcn{\kappa}{_{AB}}-\frac{1}{2}\mcn{_{AB}}\ctcn{\kappa}{}\ ,
		\end{align}
	and of $ \ct{\ell}{_{\alpha}} $,
		\begin{align}
			\ctcn{\theta}{_{AB}}&\defeq\ctcn{E}{^\alpha_{A}}\ctcn{E}{^\beta_{B}} \cd{_{\alpha}}\ct{\ell}{_{\beta}}\ ,\\
		 	\ctcn{\theta}{}&\defeq \mcn{^{AB}}\ctcn{\theta}{_{AB}}\ ,\\
		 	\ctcn{\sigma}{_{AB}}&\defeq\ctcn{\theta}{_{AB}}-\frac{1}{2}\mcn{_{AB}}\ctcn{\theta}{}\ .
		\end{align}
	Observe that by definition $ \ct{k}{^\alpha} $ is geodesic, as it is hypersurface-orthogonal,
		\begin{equation}\label{eq:k-geodesic-nu}
			\ct{k}{^{\mu}}\cd{_{\mu}}\ct{k}{^\alpha}=\nu \ct{k}{^\alpha}\ ,
		\end{equation}
	for some  function $ \nu $.	Also one has
		\begin{align}
			\lied_{\vec{k}}\mcn{_{AB}}&=2\ctcn{\kappa}{_{AB}}\ ,\\
			\lied_{\vec{k}}\mcn{^{AB}}&=-2\ctcn{\kappa}{^{AB}}\ .\label{eq:lied-invq}
		\end{align}
	Observe that the two-dimensional vorticities of both vector fields vanish by definition, however $ \ct{H}{_{[ab]}}\neq 0 $ in general. Using \cref{eq:lie-ell-k},
		\begin{align}\label{eq:H-dotF-Phi}
			\ctcn{E}{^m_{A}}\ct{k}{^a}\ct{H}{_{am}}=\ctcn{\varphi}{_{A}}-\ctcn{E}{^m_{A}}\dpart{_{m}}\ln\dot{F}\ .
		\end{align}
	Now, introduce an induced connection as 
		\begin{equation}
			\ct{v}{^a}\cds{_{a}}\ct{w}{^b}\defeq \ct{\omega}{_{\beta}^b}\ct{v}{^\alpha}\cd{_{\alpha}}\ct{w}{^\beta}\ ,\quad \forall\quad \ct{v}{^\alpha}=\ct{e}{^\alpha_{a}}\ct{v}{^a}\ ,\quad\ct{w}{^\alpha}=\ct{e}{^\alpha_{a}}\ct{w}{^a}\ .
		\end{equation}
	This is a torsion-free connection, not metric in general,  and it depends on the rigging vector, which has been fixed to the $ \ct{\ell}{_{a}} $ of \cref{eq:ell}. For any $ \ct{T}{_{\alpha}^\beta} $ defined at least on $ \N $,
		\begin{equation}
			\ct{e}{^\alpha_{a}}\ct{e}{^\beta_{b}}\ct{\omega}{_{\gamma}^c}\cd{_{\alpha}}\ct{T}{_{\beta}^\gamma}=\cds{_{a}}\cts{T}{_{b}^c}-\ct{T}{_{\mu}^c}\ct{\ell}{^\mu}\ct{K}{_{ab}}-\ct{T}{_{b}^\mu}\ct{k}{_{\mu}}\ct{\psi}{_{a}^c}\ ,
		\end{equation}
	where the notation $ \cts{T}{_{a}^b}\defeq\ct{e}{^\alpha_{a}}\ct{\omega}{_{\beta}^b} \ct{T}{_{\alpha}^\beta}$ has been introduced. In particular, for the induced metric $ \ms{_{ab}}$ one has
		\begin{equation}
			\cds{_{a}}\ms{_{bc}}=\ct{\ell}{_{b}}\ct{K}{_{ac}}+\ct{\ell}{_{c}}\ct{K}{_{ab}}\ .
		\end{equation}
	The connection is neither volume preserving, in general, and one has that $ \cds{_{a}}\ct{\epsilon}{_{bcd}}=0 \iff \ct{\varphi}{_{a}}=0$. With the connection coefficients
		\begin{equation}
			\cts{\Gamma}{^a_{bc}}=\ct{\omega}{_{\alpha}^a}\ct{e}{^{\beta}_{b}}\cd{_{\beta}}\ct{e}{^\alpha_{c}}\ ,
		\end{equation}
	one computes the associated curvature tensor $ \cts{R}{_{abc}^d} $ that satisfies the Ricci identity
		\begin{equation}
			\cts{R}{_{abc}^d}\ct{v}{_{d}}=\prn{\cds{_{a}}\cds{_{b}}-\cds{_{b}}\cds{_{a}}}\ct{v}{_{c}}\ ,
		\end{equation}
	and 
		\begin{align}
			\cts{R}{_{abc}^d}+\cts{R}{_{bac}^d}&=0\ ,\\
			\cts{R}{_{abc}^d}+\cts{R}{_{bca}^d}+\cts{R}{_{cab}^d}&=0\ ,\label{eq:permut-riemann}\\
			\cds{_{[e}}\cts{R}{_{ab]c}^d}&=0\ .\label{eq:ricci-riemann-N}
		\end{align}
	The Gauss equation reads
		\begin{equation}\label{eq:gauss-N}
			\ct{e}{^\alpha_{a}}\ct{e}{^\beta_{b}}\ct{e}{^\gamma_{c}}\ct{\omega}{_{\delta}^d}\ct{R}{_{\alpha\beta\gamma}^\delta}= \cts{R}{_{abc}^d}+\ct{K}{_{ac}}\ct{\psi}{_{b}^d}-\ct{K}{_{bc}}\ct{\psi}{_{a}^d}\ .
		\end{equation}
	Also,
		\begin{align}
			\ct{e}{^\alpha_{a}}\ct{e}{^\beta_{b}}\ct{R}{_{\alpha\beta\gamma}^\delta}\ct{\ell}{^\gamma}\ct{\omega}{_{\delta}^d}&=\cds{_{b}}\ct{\psi}{_{a}^d}-\cds{_{a}}\ct{\psi}{_{b}^d}-\ct{\varphi}{_{a}}\ct{\psi}{_{b}^d}+\ct{\varphi}{_{b}}\ct{\psi}{_{a}^d}\ ,\\
			\ct{e}{^{\alpha}_{a}}\ct{e}{^{\beta}_{b}}\ct{e}{^{\gamma}_{c}}\ct{k}{_{\delta}}\ct{R}{_{\alpha\beta\gamma}^{\delta}}&=\cds{_{[a}}\ct{K}{_{b]c}}\ ,\label{eq:codazzi-NKappa}\\
			\ct{e}{^\alpha_{a}}\ct{e}{^\gamma_{c}}\ct{\ell}{^\beta}\ct{k}{_{\delta}}\ct{R}{_{\alpha\beta\gamma}^\delta}&=2\ct{k}{^d}\cds{_{[c}}\ct{H}{_{d]a}}\ ,\\
			\ct{e}{^\alpha_{a}}\ct{e}{^\beta_{b}}\ct{R}{_{\alpha\beta\gamma}^\delta}\ct{\ell}{^\gamma}\ct{k}{_{\delta}}&=2\ct{k}{^c}\cds{_{[b}}\ct{H}{_{a]c}}=2\cds{_{[b}}\ct{\varphi}{_{a]}}+2\ct{\psi}{_{[b}^d}\ct{K}{_{a]d}}\ .
		\end{align}
	Defining the natural contraction
		\begin{equation}
			\cts{R}{_{ab}}\defeq \cts{R}{_{acb}^c}\ ,
		\end{equation}
	one has the relation
		\begin{equation}\label{eq:codazzi-N}
			\ct{e}{^\alpha_{a}}\ct{e}{^\beta_{b}}\ct{R}{_{\alpha\beta}}=\cts{R}{_{ab}}+2\ct{K}{_{b[a}}\ct{\psi}{_{c]}^c}-2\ct{k}{^d}\cds{_{[b}}\ct{H}{_{d]a}}\ .
		\end{equation}
	The last term in this expression can be expanded as
		\begin{equation}\label{eq:combination-H}
			2\ct{k}{^{d}}\cds{_{[b}}\ct{H}{_{d]a}}=-\ct{k}{^{d}}\ms{_{ma}}\prn{2\cds{_{[d}}\ct{\psi}{_{b]}^{m}}-\ct{\psi_{[d}}{^{m}}\ct{\varphi}{_{b]}}}+\ct{\ell}{_{a}}\ct{k}{^{d}}\brkt{2\cds{_{[d}}\ct{\varphi}{_{b]}}+\ct{K}{_{bm}}\ct{\psi}{_{d}^{m}}}\ .
		\end{equation}
	It is important to observe that, although the right hand side of \cref{eq:codazzi-N} is symmetric, in general the tensor $ \cts{R}{_{ab}} $ is not. Also,
		\begin{equation}
			\cts{R}{_{abd}^d}=2\ct{k}{^c}\cds{_{[b}}\ct{H}{_{a]c}}+2\ct{K}{_{d[b}}\ct{\psi}{_{a]}^d}=2\cds{_{[b}}\ct{\varphi}{_{a]}}\ 
		\end{equation}
	is different from zero too, in general. Indeed, using \cref{eq:permut-riemann},
		\begin{equation}\label{eq:antisym-Rab}
			\cts{R}{_{[ab]}}=\cds{_{[b}}\ct{\varphi}{_{a]}}\ .
		\end{equation}
	With the symmetric part of $ \cts{R}{_{ab}} $, define the the scalar
		\begin{equation}
			\cts{R}{}\defeq\ms{^{ab}}\cts{R}{_{ab}}\ .
		\end{equation}

	Apart from the induced connection on $ \N $, one can    define  a (one-parameter family of) two-dimensional connection(s)
		\begin{equation}
			\ctcn{\gamma}{^A_{BC}}=\ctcn{W}{_{a}^A}\ctcn{E}{^{b}_B}\cds{_{b}}\ctcn{E}{^a_{C}}\ ,
		\end{equation}
	yielding the following covariant derivative
		\begin{equation}
			\ct{v}{^A}\cdcn{_{A}}\ct{w}{^B}\defeq \ct{W}{_{b}^B}\ct{v}{^a}\cd{_{a}}\ct{w}{^b}\ ,\quad \forall\quad \ct{v}{^a}=\ct{E}{^a_{A}}\ct{v}{^A}\ ,\quad\ct{w}{^a}=\ct{E}{^a_{A}}\ct{w}{^A}\ \ ,
		\end{equation}
	that satisfies 
		\begin{align}
			\ctcn{E}{^\alpha_{A}}\ctcn{E}{^\gamma_{C}}\ctcn{W}{_{\beta}^{B}}\cd{_{\alpha}}\ct{T}{^{\beta}_{\gamma}}=&\cdcn{_{A}}\ctcn{T}{^B_{C}}-\ct{T}{^{\beta}_{\gamma}}\ct{k}{_{\beta}}\ct{E}{^{\gamma}_{C}}\ctcn{\theta}{_{A}^{B}}-\ctcn{W}{_{\beta}^{B}}\ct{T}{^{\beta}_{\gamma}}\ct{\ell}{^\gamma}\ctcn{\kappa}{_{AC}}\nonumber\\
			&-\ct{T}{^{\beta}_{\gamma}}\ct{\ell}{_{\beta}}\ctcn{E}{^{\gamma}_{C}}\ctcn{\kappa}{_{A}^{B}}-\ctcn{W}{_{\beta}^{B}}\ct{T}{^{\beta}_{\gamma}}\ct{k}{^{\gamma}}\ctcn{\theta}{_{AC}}\ ,\\
			\ctcn{E}{^a_{A}}\ctcn{E}{^{c}_{C}}\ctcn{W}{_{b}^{B}}\cds{_{a}}\ct{T}{^{b}_{c}}=&\cdcn{_{A}}\ctcn{T}{^{B}_{C}}-\ct{T}{^{b}_{c}}\ct{\ell}{_{b}}\ctcn{E}{^{c}_{C}}\ctcn{\kappa}{_{A}^{B}}-\ctcn{W}{_{b}^{B}}\ct{T}{^{b}_{c}}\ct{k}{^{c}}\ctcn{\theta}{_{AC}}\ .
		\end{align}
	To derive these expressions one uses the relation
		\begin{align}
			\ms{^{mb}}\ct{e}{^{\mu}_{m}}\ct{g}{_{\mu\alpha}}&=\ct{\omega}{_{\alpha}^{b}}+\ct{\ell}{_{\alpha}}\ct{k}{^{b}}\ .
		\end{align}
	This connection is metric, in the sense that
		\begin{equation}
			\cdcn{_{A}}\mcn{_{BC}}=0\ ,
		\end{equation}
	and it is also torsion-free. Then, when it is particularised to a single leaf, it gives the Levi-Civita connection of the two-dimensional Riemannian manifold $ \prn{\Sc_{C},\mc{_{AB}}} $ with $ \mc{_{AB}}\defeq \mcn{_{AB}}\evalat{{F=C}} $. The corresponding curvature tensor satisfies
		\begin{align}
			\ctcn{R}{_{ABC}^{D}}\ct{v}{_{D}}&=\prn{\cdcn{_{A}}\cdcn{_{B}}-\cdcn{_{B}}\cdcn{_{A}}}\ct{v}{_{C}}\ ,\\
				\ctcn{R}{_{AC}}&\defeq \ctcn{R}{_{ADC}^D}=\ctcn{R}{_{CA}}\ .
		\end{align}
	And because these tensors are 2-dimensional
		\begin{align}
			\ctcn{R}{_{ABC}^{D}}&=\ctcn{K}{}\prn{\mcn{_{AC}}\delta_{B}^D-\delta^D_{A}\mcn{_{BC}}}\ ,\\
			\ctcn{R}{_{AB}}&=\ctcn{K}{}\mcn{_{AB}}\ ,\\
			\ctcn{R}{}&\defeq \mcn{^{AB}}\ctcn{R}{_{AB}}=2\ctcn{K}{}\label{eq:gaussian-curvature}\ ,
		\end{align}
	where $ \ctcn{K}{} $ coincides with the Gaussian curvature on each of the leaves of the foliation. The Gauss equation in this case reads
		\begin{equation}\label{eq:gauss-Foliation}
			\ctcn{E}{^{a}_{A}}\ctcn{E}{^{b}_{B}}\ctcn{E}{^{c}_{C}}\cts{R}{_{abc}^d}\ctcn{W}{_{d}^{D}}=\ctcn{R}{_{ABC}^{D}}+2\ctcn{\theta}{_{C[A}}\ctcn{\kappa}{_{B]}^{D}}\ .
		\end{equation}	
	Contracting indices $ D $ and $ B $,
		\begin{equation}
			\cts{R}{_{abc}^{d}}\ctcn{P}{^{b}_{d}}\ctcn{E}{^{a}_{A}}\ctcn{E}{^{c}_{C}}=\ctcn{R}{_{AC}}+2\ctcn{\theta}{_{C[A}}\ctcn{\kappa}{_{B]}^{B}}\ .
		\end{equation}	
	After working out the left-hand side, using \cref{eq:projector-leaves}, one arrives at
		\begin{align}
			\cts{R}{_{ac}}\ctcn{E}{^{a}_{A}}\ctcn{E}{^{c}_{C}}=&\ctcn{R}{_{AC}}+\ctcn{\theta}{_{CA}}\ctcn{\kappa}{}-2\ctcn{\theta}{_{D(A}}\ctcn{\kappa}{_{C)}^{D}}+\ct{k}{^{b}}\ct{\varphi}{_{b}}\ctcn{\theta}{_{AC}}+\lied_{\vec{k}}\ctcn{\theta}{_{AC}}\nonumber\\
			&-\cdcn{_{A}}\ctcn{\varphi}{_{C}}-\ctcn{\varphi}{_{A}}\ctcn{\varphi}{_{C}}+2\ctcn{\varphi}{_{(A}}\cdcn{_{C)}}\ln\dot{F}-\cdcn{_{C}}\prn{\ln \dot{F}}\cdcn{_{A}}\prn{\ln \dot{F}}\nonumber\\
			&+\cdcn{_{A}}\cdcn{_{C}}\ln \dot{F}\ ,
		\end{align}
	where \cref{eq:H-dotF-Phi} was used too. Notice that, using $ \ctcn{R}{_{AC}}=\ctcn{R}{_{CA}} $, one shows that $ \cts{R}{_{ac}}\ctcn{E}{^a_{[A}}\ctcn{E}{^{c}_{C]}}=-\cdcn{_{[A}}\ctcn{\varphi}{_{C]}} $, as it has to be ---recall \cref{eq:antisym-Rab}. One can decompose the Lie derivative of $ \theta_{AB} $ to get another version of the equation containing explicitly the evolution along $ \ct{k}{^b} $ of the shear $ \ctcn{\sigma}{_{AB}} $ ,
		\begin{align}
				\cts{R}{_{ac}}\ctcn{E}{^{a}_{A}}\ctcn{E}{^{c}_{C}}=&\ctcn{R}{_{AC}}+\ctcn{\theta}{_{CA}}\ctcn{\kappa}{}+\theta\ctcn{\kappa}{_{AC}}-2\ctcn{\theta}{_{D(A}}\ctcn{\kappa}{_{C)}^{D}}+\frac{1}{2}\lied_{\vec{k}}\ctcn{\theta}{}\mcn{_{AB}}+\lied_{\vec{k}}\ctcn{\sigma}{_{AC}}\nonumber\\
				&+\ct{k}{^{b}}\ct{\varphi}{_{b}}\ctcn{\theta}{_{AC}}-\cdcn{_{A}}\ctcn{\varphi}{_{C}}-\ctcn{\varphi}{_{A}}\ctcn{\varphi}{_{C}}+2\ctcn{\varphi}{_{(A}}\cdcn{_{C)}}\ln\dot{F}-\cdcn{_{C}}\prn{\ln \dot{F}}\cdcn{_{A}}\prn{\ln \dot{F}}\nonumber\\
				&+\cdcn{_{A}}\cdcn{_{C}}\ln \dot{F}\ .
			\end{align}
	Taking the trace, on gets
		\begin{align}
			\cts{R}{}=&\ctcn{R}{}+\ctcn{\theta}{}\ctcn{\kappa}{}+\lied_{\vec{k}}\ctcn{\theta}{}+\ctcn{\theta}{}\ct{k}{^b}\ct{\varphi}{_{b}}-\cdcn{_{A}}\ctcn{\varphi}{^{A}}-\ctcn{\varphi}{_{C}}\ctcn{\varphi}{^{C}}+2\ctcn{\varphi}{^{C}}\cdcn{_{C}}\ln \dot{F}\nonumber\\
			&-\cdcn{_{C}} \prn{\ln\dot{F}}\cdcn{^{C}}\prn{\ln\dot{F}}+\cdcn{_{C}}\cdcn{^{C}}\ln \dot{F}\ .
		\end{align}
	From this, one can give the traceless combination
		\begin{align}\label{eq:traceless-R}
			\ctcn{E}{^{a}_{A}}\ctcn{E}{^{c}_{C}}\cts{R}{_{ab}}-\frac{1}{2}\mcn{_{AC}}	\cts{R}{}=&\lied_{\vec{k}}\ctcn{\sigma}{_{AB}}-2\ctcn{\sigma}{_{D(A}}\ctcn{\nu}{_{C)}^{D}}+\ct{k}{^m}\ct{\varphi}{_{m}}\ctcn{\sigma}{_{AC}}+\ctcn{F}{_{AB}}\nonumber\\
			&+2\ctcn{\varphi}{_{(A}}\cdcn{_{C)}}\ln\dot{F}-\cdcn{_{A}}\ctcn{\varphi}{_{C}}-\ctcn{\varphi}{_{A}}\ctcn{\varphi}{_{C}}\nonumber\\
			&-\frac{1}{2}\mcn{_{AC}}\prn{2\ctcn{\varphi}{^M}\cdcn{_M}\ln\dot{F}-\cdcn{_{M}}\ctcn{\varphi}{^{M}}-\ctcn{\varphi}{_{M}}\ctcn{\varphi}{^{M}}}\ .
		\end{align}
	where the traceless, symmetric tensor
		\begin{equation}\label{eq:combination-F}
			\ctcn{F}{_{AC}}\defeq \cdcn{_{A}}\cdcn{_{C}}\ln\dot{F}-\cdcn{_{A}}\prn{\ln\dot{F}}\cdcn{_{C}}\prn{\ln\dot{F}}-\frac{1}{2}\mcn{_{AC}}\brkt{\cdcn{_{M}}\cdcn{^{M}}\ln\dot{F}-\cdcn{_{M}}\prn{\ln\dot{F}}\cdcn{^{M}}\prn{\ln\dot{F}}}
		\end{equation}
	has been introduced. Observe that \cref{eq:traceless-R} only contains the symmetric part of $ \ctcn{E}{^{a}_{A}}\ctcn{E}{^{c}_{C}}\cts{R}{_{ac}} $ ---recall \cref{eq:antisym-Rab}.	Also, taking traces and using the formulae for the projectors one gets the relation
	\begin{equation}
		2\ct{k}{^\mu}\ct{\ell}{^{\nu}}\ct{R}{_{\mu\nu}}=2\cts{R}{}-\ctcn{R}{}-\ct{R}{}\ .
	\end{equation}
	
	Next, define
		\begin{equation}
			\ct{S}{_{\alpha\beta}}\defeq \ct{R}{_{\alpha\beta}}-\frac{1}{6}\ct{g}{_{\alpha\beta}}\ct{R}{}\ ,
		\end{equation}
	which is (twice) the space-time Schouten tensor. Its pullback to $ \N $ and to the leaves are defined as
		\begin{align}
			\cts{S}{_{ab}}\defeq&\frac{1}{2}\ct{e}{^{\alpha}_{a}}\ct{e}{^{\beta}_{b}}\ct{S}{_{\alpha\beta}}\ ,\\
			\ctcn{S}{_{AB}}\defeq&\ctcn{E}{^{a}_{A}}\ctcn{E}{^{b}_{B}}\cts{S}{_{ab}}\ .\label{eq:schouten-2-dim}
		\end{align}
	Using the Gauss \eqref{eq:gauss-N} and Codazzi \eqref{eq:codazzi-N} relations above, it follows that
		\begin{equation}
		2\cts{S}{_{ab}}=\cts{R}{_{ab}}-\frac{1}{2}\ms{_{ab}}\cts{R}{}+2\ct{K}{_{b[a}}\ct{\psi}{_{d]}^{d}}-2\ct{k}{^{d}}\cds{_{[b}}\ct{H}{_{d]a}}+\frac{1}{2}\ms{_{ab}}\prn{\frac{1}{2}\ctcn{R}{}-\ct{\ell}{^{\nu}}\ct{k}{^{\mu}}\ct{S}{_{\mu\nu}}}\ .
		\end{equation}
	Projecting this last equation to the leaves and using \cref{eq:combination-H} one obtains
		\begin{align}
			2\ctcn{S}{_{AB}}=&\ctcn{E}{^{a}_{A}}\ctcn{E}{^{b}_{B}}\prn{\cts{R}{_{ab}}-\frac{1}{2}\ms{_{ab}}\cts{R}{}}+\lied_{\vec{k}}\ctcn{\sigma}{_{AB}}-2\ctcn{\nu}{_{C(A}}\ctcn{\sigma}{_{B)}^{C}}+\ctcn{\nu}{_{AB}}\ctcn{\theta}{}-\ctcn{\sigma}{_{AB}}\ctcn{\kappa}{}+\ctcn{\theta}{_{AB}}\ct{k}{^d}\ct{\varphi}{_{d}}\nonumber\\
			&-\ctcn{\varphi}{_{A}}\ctcn{\varphi}{_{B}}-\cdcn{_{(B}}\ctcn{\varphi}{_{A)}}+2\ctcn{\varphi}{_{(A}}\cdcn{_{B)}}\ln\dot{F}-\cdcn{_{B}}\prn{\ln\dot{F}}\cdc{_{A}}\prn{\ln\dot{F}}+\cdcn{_{B}}\cdcn{_{A}}\ln\dot{F}\nonumber\\
			&+\frac{1}{2}\mcn{_{AB}}\prn{\frac{1}{2}\ctcn{R}{}+\ct{S}{_{\mu\nu}}\ct{k}{^\mu}\ct{\ell}{^\nu}+\ctcn{\kappa}{}\ctcn{\theta}{}+\lied_{\vec{k}}\ctcn{\theta}{}}\ .
		\end{align}
	As it will be shown later, the traceless part of $ \ctcn{S}{_{AB}} $ plays an important role in the analysis of gravitational radiation. Using the last equation together with \cref{eq:traceless-R},
		\begin{align}
			\ctcn{S}{_{AB}}-\frac{1}{2}\mcn{_{AB}}\ctcn{S}{^{M}_{M}}=&\lied_{\vec{k}}\ctcn{\sigma}{_{AB}}+\frac{1}{2}\ctcn{\nu}{_{AB}}\ctcn{\theta}{}-\frac{1}{2}\ctcn{\sigma}{_{AB}}\ctcn{\kappa}{}-2\ctcn{\nu}{_{C(A}}\ctcn{\sigma}{_{B)}^{C}}+\ct{k}{^d}\ct{\varphi}{_{d}}\ctcn{\sigma}{_{AB}}\nonumber\\
			&-\ctcn{\varphi}{_{A}}\ctcn{\varphi}{_{B}}-\cdcn{_{(B}}\ctcn{\varphi}{_{A)}}+2\ctcn{\varphi}{_{(A}}\cdcn{_{B)}}\ln\dot{F}-\frac{1}{2}\mcn{_{AB}}\prn{2\ctcn{\varphi}{^{M}}\cdcn{_{M}}\ln\dot{F}-\ctcn{\varphi}{_{C}}\ctcn{\varphi}{^{C}}-\cdcn{_{C}}\ctcn{\varphi}{^{C}}}\nonumber\\
			&+\ct{F}{_{AB}}\ ,\label{eq:schouten-shear}
		\end{align}
	where the combination $ \ct{F}{_{AB}} $ depends on the choice of foliation and is defined in \cref{eq:combination-F}. The standard decomposition of the Riemann tensor in terms of the Weyl tensor $ \ct{C}{_{\alpha\beta\gamma}^{\delta}} $ and Ricci tensor can be written as
		\begin{equation}\label{eq:decomposition-Riemann}
			\ct{R}{_{\alpha\beta\gamma}^{\delta}}=\ct{C}{_{\alpha\beta\gamma}^{\delta}}+\ct{g}{_{\gamma[\alpha}}\ct{S}{_{\beta]}^{\delta}}+\delta^{\delta}_{[\beta}\ct{S}{_{\alpha]\gamma}}\ .
		\end{equation}
	Contract now with $ \ctcn{P}{^{\alpha\gamma}}\ctcn{E}{^{\beta}_{B}}\ct{\ell}{_{d}} $ and use \cref{eq:gauss-N} to substitute for the left-hand side; the result is
		\begin{equation}\label{eq:Sell}
			\ctlcn{S}{_{B}}=2\mcn{^{AC}}\cdcn{_{[A}}\ctcn{\theta}{_{B]C}}-\ctcn{\theta}{_{C[A}}\ctcn{\varphi}{_{B]}}\mcn{^{CA}}+\ctlcn{d}{_{B}}
		\end{equation}
	where one defines
		\begin{align}
			\ctlcn{S}{_{A}}&\defeq\frac{1}{2}\ct{\ell}{^\beta}\ctcn{E}{^\alpha_{A}}\ct{S}{_{\alpha\beta}}\ ,\\
			\ctlcn{d}{_{A}}&\defeq \ct{\ell}{^\mu}\ct{\ell}{^{\nu}}\ct{k}{_\rho}\ct{E}{^{\alpha}_{A}}\ct{C}{_{\mu\alpha\nu}^{\rho}}\ .
		\end{align}
	Similarly, using \cref{eq:codazzi-NKappa},	
		\begin{equation}\label{eq:SN}
			\ctkcn{S}{_{A}}=\cdcn{_{C}}\ctcn{\kappa}{_{B}^{C}}-\cdcn{_{B}}\ctcn{\kappa}{}+\ctkcn{d}{_{A}}
		\end{equation}	
	where now
		\begin{align}
			\ctkcn{S}{_{A}}&\defeq\frac{1}{2}\ct{k}{^\beta}\ctcn{E}{^\alpha_{A}}\ct{S}{_{\alpha\beta}}\ ,\\
			\ctkcn{d}{_{A}}&\defeq \ct{k}{^\mu}\ct{k}{^{\nu}}\ct{\ell}{_\rho}\ct{E}{^{\alpha}_{A}}\ct{C}{_{\mu\alpha\nu}^{\rho}}\ .
		\end{align}
	Also, one can contract \cref{eq:decomposition-Riemann} with $\ctcn{E}{^{\alpha}_{A}}\ctcn{E}{^{\beta}_{B}}\ctcn{E}{^{\gamma}_{C}} \ctcn{W}{_{\delta}^{D}} $ and then take the traces with $ \mcn{^{AC}}\delta_{D}^{B} $ which, using \cref{eq:gauss-N,eq:gauss-Foliation} on the right-hand side, yields
		\begin{equation}
			2\ctcn{S}{_{M}^{M}}=-\ct{C}{_{AB}^{AB}}+\ctcn{R}{}+\ctcn{\theta}{}\ctcn{\kappa}{}-2\ctcn{\nu}{_{AC}}\ctcn{\theta}{^{AC}}\ .\label{eq:schouten-trace}
		\end{equation}
	The first scalar on the right-hand side is defined as
		\begin{equation}
			\ct{C}{_{AB}^{AB}}\defeq \ctcn{P}{_{\delta}^\beta}\ctcn{P}{^{\alpha\gamma}}\ctcn{C}{_{\alpha\beta\gamma}^{\delta}}
		\end{equation}
	and carries information about the Coulomb part of the Weyl tensor\footnote{In the null tetrad containing $ \cbrkt{\ct{\ell}{^{\alpha}},\ct{k}{^\alpha}} $, this is $ -4\Re\prn{\Psi_{2}} $.}. Also, $ \ctlcn{d}{_{A}} $ and $ \ctkcn{d}{_{A}} $ carry part of the radiative information of the Weyl tensor\footnote{They are, respectively, $ \Psi_1 $ and $ \Psi_3 $.}.
	\subsection{Identification of the news tensor on a null hypersurface}
	 In this section, \emph{conformal rescalings} acting as
	 	\begin{equation}
	 		\cta{g}{_{\alpha\beta}}=\lambda^2\ct{g}{_{\alpha\beta}}
	 	\end{equation}
	 and
	 	\begin{equation}\label{eq:conformal-mc}
		 \prn{\cta{k}{^{a}},\ctsa{g}{_{ab}}}=\prn{\lambda^{-1}\ct{k}{^{a}},\lambda^{2}\ms{_{ab}}} \ ,
	 	\end{equation}
	 are considered, where $ \lambda>0 $ is a smooth function. The motivation is that at conformal infinity with $ \Lambda=0 $, \cref{eq:conformal-mc} is precisely the gauge conformal freedom (setting $ \lambda=\omega $, see the notation in \cref{sec:conf-space}) as will be remarked in \cref{sec:conf-space}. So far, however, only a null hypersurface foliated by topological two-spheres is needed. Under \cref{eq:conformal-mc}, the two-dimensional projected Schouten tensor \eqref{eq:schouten-2-dim} transforms as
	 	\begin{align}
	 	\ctcna{S}{_{AB}}=&\ctcn{S}{_{AB}}-\lambda^{-1}\cdcn{_{A}}\ctcn{\lambda}{_{B}}+2\lambda^{-2}\ctcn{\lambda}{_{A}}\ctcn{\lambda}{_{B}}-\frac{1}{2}\lambda^{-2}\mcn{_{AB}}\ctcn{\lambda}{_{M}}\ctcn{\lambda}{^{M}}\nonumber\\ 
	 	&+ \lambda^{-1}\prn{\ctcn{\theta}{_{AB}}\ct{k}{^{\mu}}\ct{\lambda}{_{\mu}}+\ctcn{\kappa}{_{AB}}\ct{\ell}{^{\mu}}\ct{\lambda}{_{\mu}}}+\lambda^{-2}\mcn{_{AB}}\ct{k}{^\mu}\ct{\lambda}{_{\mu}}\ct{\ell}{^{\nu}}\ct{\lambda}{_{\nu}}\ ,	 	
	 	\end{align}
	 where $ \ct{\lambda}{_{\mu}}\defeq\cd{_{\mu}}\lambda $ and $ \ctcn{\lambda}{_{A}}\defeq \cdcn{_{A}}\lambda $.
	 Then the tensor field
	 	\begin{equation}\label{eq:U-definition}
				\ctcn{U}{_{AB}}\defeq\ct{S}{_{AB}}+\frac{1}{2}\mcn{_{AB}}\prn{\ctcn{\nu}{_{MC}}\ctcn{\sigma}{^{MC}}-\frac{1}{2}\ctcn{\kappa}{}\ctcn{\theta}{}+\frac{1}{2}\ctcn{C}{_{MD}^{MD}}}-\frac{1}{2}\prn{\theta\ctcn{\nu}{_{AB}}+\ctcn{\kappa}{}\ctcn{\sigma}{_{AB}}}
		\end{equation}
	has the following properties
		\begin{align}
			\ctcn{U}{_{M}^M}&=\ctcn{K}{}\ \label{eq:U-trace},\\
			\ctcna{U}{_{AB}}&=\ctcn{U}{_{AB}}-\lambda^{-1}\cdcn{_{A}}\ctcn{\lambda}{_{B}}+2\lambda^{-2}\ctcn{\lambda}{_{A}}\ctcn{\lambda}{_{B}}-\frac{1}{2}\lambda^{-2}\mcn{_{AB}}\ctcn{\lambda}{_{M}}\ctcn{\lambda}{^{M}}\ .\label{eq:conformal-U}
		\end{align}
	This particular conformal behaviour and trace are suitable for application of the following theorems adapted from \cite{Fernandez-AlvarezSenovilla2022b}
		\begin{thm}[The tensor $ \rho $ \cite{Fernandez-AlvarezSenovilla2022b}]\label{thm:rho-tensor}
		If a leaf $ \Sc_{C} $ has $ \mathbb{S}^2 $-topology, there is a unique symmetric tensor field $ \ct{\rho}{_{AB}} $  whose behaviour under conformal rescalings \eqref{eq:conformal-mc} is as in \eqref{eq:conformal-U} and satisfies the equation
			\begin{equation}\label{eq:rho-diff-eq}
				\cdc{_{[C}}\ct{\rho}{_{A]B}}=0
			\end{equation}
		in any conformal frame. This tensor field must have a trace $ \ct{\rho}{^E_E}=aK $ and obeys  
			\begin{equation}\label{eq:rho-lie-CKVF}
				\lied_{\vec{\chi}}\ct{\rho}{_{AB}}=-a\cdc{_{A}}\cdc{_{B}}\phi
			\end{equation}
		independently of the conformal frame, where $ \ct{\chi}{^A} $ is any CKVF of $ \prn{\Sc_C,\mc{_{AB}}} $ and $ \phi\defeq\cdc{_M}\ct{\chi}{^M}/2 $. Specifically, it is invariant under transformations generated by KVF (and homothetic Killing vectors) of $ \prn{\Sc_C,\mc{_{AB}}} $. Furthermore, it is given for round spheres by $ \ct{\rho}{_{AB}}=\mc{_{AB}}aK/2 $.
		\end{thm}	
	Readily, the generalisation to a tensor field $ \ctcn{\rho}{_{ab}} $ on $ \N $ follows,
		\begin{corollary}[The tensor field $ \rho $ for a foliated $\N$ with $\mathbb{S}^2$ leaves]\label{thm:rho-tensor-foliation}
			Assume $\N$ is foliated by leaves $ \Sc_{C} $ with $ \mathbb{S}^2 $-topology, with foliation function $ F $ as in \cref{eq:foliation-F}. Then, there is a unique tensor field $ \ctcn{\rho}{_{ab}} $  on $ \N $ orthogonal to $ \ct{k}{^a} $ (equivalently, a one-parameter family of symmetric tensor fields $ \ctcn{\rho}{_{AB}}\prn{F}\defeq \ctcn{E}{^a_A}\ctcn{E}{^b_B}\ctcn{\rho}{_{ab}} $ ) whose behaviour under conformal rescalings \eqref{eq:conformal-mc} is as in \cref{eq:conformal-U} and satisfies the equation
			 	\begin{equation}\label{eq:rho-diff-eq-foliation}
				\ctcn{P}{^d_a}\ctcn{P}{^e_b}\ctcn{P}{^f_c}\cds{_{[f}}\ctcn{\rho}{_{d]e}}= 0\ 
				\end{equation}
			in any conformal frame. This tensor field must have a trace $ \ctcn{\rho}{^e_e}\defeq \ctcn{P}{^{ae}}\ctcn{\rho}{_{ae}}=a\cscn{K} $ and reduces, at each leaf, to the corresponding tensor of \cref{thm:rho-tensor} with all its properties.
		\end{corollary}
		\begin{proof}
			The proof follows the same lines of that one in corollary 6.2 of \cite{Fernandez-AlvarezSenovilla2022b}. The pullback of \cref{eq:rho-diff-eq-foliation} to a leaf $ \Sc_{C} $ gives precisely \cref{eq:rho-diff-eq} on that leaf. Then one can define $ \ctcn{\rho}{_{ab}} $ at each leaf as
				\begin{equation}
					\ctcn{\rho}{_{ab}}\evalat{\Sc_{C}}=\prn{\ctcn{W}{_{a}^{A}}\ctcn{W}{_{b}^{B}}\ct{\rho}{_{AB}}}\evalat{\Sc_{C}}\ .
				\end{equation}
			Noting \cref{eq:foliation-cap,eq:foliation-cup}, the unique solution of \cref{eq:rho-diff-eq-foliation} is given by
				\begin{equation}
					\ctcn{\rho}{_{ab}}=\ctcn{W}{_{a}^A}\ctcn{W}{_{b}^{B}}\ct{\rho}{_{AB}}\evalat{\Sc_{F}}\ ,\quad\ct{k}{^d}\ctcn{\rho}{_{ad}}=0 \ .
				\end{equation}
			The function $ \ctcn{K}{} $ \eqref{eq:gaussian-curvature} coincides with the Gaussian curvature $ \ct{K}{} $ of each leaf $ \Sc $, and hence one has that $ \ctcn{\rho}{_{ab}}\ctcn{P}{^{ab}}=a\ctcn{K}{} $ and that for round (family of) metrics $ \mcn{_{AB}} $, $ 2\ctcn{\rho}{_{AB}}=a\ctcn{K}{}\mcn{_{AB}} $\ .
		\end{proof}
		These results and the definition \eqref{eq:U-definition} of $ \ctcn{U}{_{AB}} $ serve to define a \emph{news tensor} on $ \N $:
		\begin{thm}[News tensor for a foliated null hypersurface]\label{thm:news-N}
			Let $ \N $ be a null hypersurface foliated by leaves with $ \mathbb{S}^2 $-topology and with $ F $ the defining function \eqref{eq:foliation-F}. Then, there is a one-parameter family (depending on $ F $) of symmetric, traceless, conformal-invariant tensor fields
			\begin{equation}\label{eq:News}
			\ctcn{N}{_{AB}}\defeq \ctcn{U}{_{AB}}-\ctcn{\rho}{_{AB}}\spacef,
			\end{equation}
			that satisfies the conformal-invariant equation
			\begin{equation}\label{eq:diffUAB-foliations}
			\cdcn{_{[A}}\ctcn{U}{_{B]C}}=\cdcn{_{[A}}\ctcn{N}{_{B]C}}\spacef,
			\end{equation}
			where $ \ctcn{\rho}{_{AB}} $ is the family of tensor fields of \cref{thm:rho-tensor-foliation} (for $ a=1 $). Besides, $ \ctcn{N}{_{AB}} $ is unique with these properties.
		\end{thm}
		\begin{proof}
			The one-parameter family of tensor fields $ \ctcn{N}{_{AB}} $ is symmetric, traceless and conformally invariant as a consequence of \cref{eq:U-definition,eq:conformal-U,eq:U-trace} and \cref{thm:rho-tensor-foliation}. The uniqueness of $ \ctcn{N}{_{AB}} $ follows from \cref{thm:rho-tensor-foliation} too and \cref{eq:diffUAB-foliations}.
		\end{proof}
		Thus, the role of $ \ctcn{\rho}{_{AB}} $ is to remove the trace of $ \ctcn{U}{_{AB}} $ and to make it conformally invariant. Observe that an expression for the resulting news tensor $ \ctcn{N}{_{AB}} $, using \cref{eq:schouten-shear,eq:schouten-trace} and decomposition \cref{eq:U-definition}, can be obtained,
			\begin{align}
				\ctcn{N}{_{AB}}&=\lied_{\vec{N}}\ctcn{\sigma}{_{AB}}-\ctcn{\sigma}{_{AB}}\ctcn{\kappa}{}-2\ctcn{\nu}{_{C(A}}\ctcn{\sigma}{_{B)}^{C}}-\ctcn{\rho}{_{AB}}+\frac{1}{2}\mcn{_{AB}}\ctcn{K}{}\nonumber\\
				&+\ct{k}{^d}\ct{\varphi}{_{d}}\ctcn{\sigma}{_{AB}}-\ctcn{\varphi}{_{A}}\ctcn{\varphi}{_{B}}-\cdcn{_{(B}}\ctcn{\varphi}{_{A)}}+2\ctcn{\varphi}{_{(A}}\cdcn{_{B)}}\ln\dot{F}-\frac{1}{2}\mcn{_{AB}}\prn{2\ctcn{\varphi}{^{M}}\cdcn{_{M}}\ln\dot{F}-\ctcn{\varphi}{_{C}}\ctcn{\varphi}{^{C}}-\cdcn{_{C}}\ctcn{\varphi}{^{C}}}\nonumber\\
				&+\ct{F}{_{AB}}\ ,\label{eq:N-lied-sigma}
			\end{align}
		where $ \ct{F}{_{AB}} $ was defined in \cref{eq:combination-F}.
\section{Radiation on null hypersurfaces in conformal space-time}\label{sec:conf-space}
Now let $ \prn{M,\ct{g}{_{\alpha\beta}}} $ be the conformal completion of a physical space-time ---see \cref{sec:conventions}. In conformal space-time the dynamics are governed by 
		\begin{align}
						&\cd{_\alpha}\ct{N}{_\beta} = -\frac{1}{2}\Omega\ct{S}{_{\alpha\beta}}+ \cs{f}\ct{g}{_{\alpha\beta}}+ \frac{1}{2}\Omega^2\varkappa\ctt{{T}}{_{\alpha\beta}}\spacef,\label{eq:cefesDerN}\\
						&\ct{N}{_\mu}\ct{N}{^\mu}=\frac{\Omega^3}{12}\varkappa \cs{T}-\frac{\Lambda}{3}+ 2\Omega f\spacef,\label{eq:cefesNormN}\\
						&\cd{_\alpha}f= -\frac{1}{2}\ct{S}{_{\alpha\mu}}\ct{N}{^\mu}+\frac{1}{2}\Omega\varkappa\ct{N}{^\mu}\ctt{{T}}{_{\alpha\mu}}-\frac{1}{24}\Omega^2\varkappa\cd{_\alpha}\cs{T}- \frac{1}{8}\Omega\varkappa\ct{N}{_\alpha}\cs{T}\spacef,\label{eq:cefesDerF}\\
						&\ct{d}{_{\alpha\beta\gamma}^\mu}\ct{N}{_\mu}+\cd{_{[\alpha}}\prn{\ct{S}{_{\beta]\gamma}}}-\Omega\ct{y}{_{\alpha\beta\gamma}}=0\spacef,\label{eq:cefesDerSchouten}\\
						&\ct{y}{_{\alpha\beta\gamma}}+\cd{_{\mu}}\ct{d}{_{\alpha\beta\gamma}^\mu}= 0\spacef,\label{eq:cefesDerWeyl}\\
						&\ct{R}{_{\alpha\beta\gamma\delta}} =\Omega\ct{d}{_{\alpha\beta\gamma\delta}}+ \ct{g}{_{\alpha[\gamma}}\ct{S}{_{\delta]\beta}}-\ct{g}{_{\beta[\gamma}}\ct{S}{_{\delta]\alpha}}\spacef.	\label{eq:cefesRiemann}
						\end{align}
		The variables entering these equations are the  Schouten tensor $ \ct{S}{_{\alpha\beta}} $, the rescaled Weyl tensor
			\begin{equation}
				\ct{d}{_{\alpha\beta\gamma}^{\delta}}\defeq\frac{1}{\Omega}\ct{C}{_{\alpha\beta\gamma}^{\delta}}\ ,
			\end{equation}
		 the Friedrich scalar $ f $
			\begin{equation}
				f\defeq \frac{1}{4}\cd{_{\mu}}\ct{N}{^{\mu}}+\frac{\Omega}{24}\ct{R}{}\ ,
			\end{equation}
		and the rescaled Cotton-York tensor
			\begin{equation}
				\ct{y}{_{\alpha\beta\gamma}}\defeq\frac{1}{\Omega}\pd{_{[\alpha}}\pt{S}{_{\beta]\gamma}}\ ,
			\end{equation}
		where $ \pd{_{\alpha}} $ and $ \pt{S}{_{\alpha\beta}} $  stand for the physical covariant derivative and physical Schouten tensor, respectively. Also, for convenience, the following definition was introduced,
	\begin{equation}
		\ctt{T}{_{\alpha\beta}}\defeq\ct{T}{_{\alpha\beta}}-\frac{1}{4}\ct{T}{}\ct{g}{_{\alpha\beta}}\ ,
	\end{equation}
and in vacuum  $ \ct{T}{_{\alpha\beta}}=0=\ct{y}{_{\alpha\beta\gamma}} $. These equations are the so called conformal Einstein field equations (CEFE) \cite{Friedrich1981a,Friedrich1981b}; see also \cite{Paetz2013}, or \cite{Fernandez-AlvarezSenovilla2022a} for a detailed derivation with the present conventions.  Evaluation of the CEFE at $ \scri $  yields
						\begin{align}
						&\cd{_\alpha}\ct{N}{_\beta} \eqs f\ct{g}{_{\alpha\beta}}\spacef,\label{eq:cefesScriDerN}\\
						&\ct{N}{_\mu}\ct{N}{^\mu}\eqs -\frac{\Lambda}{3} \spacef,\label{eq:cefesScriNormN}\\
						&\cd{_\alpha}f\eqs -\frac{1}{2}\ct{S}{_{\alpha\mu}}\ct{N}{^\mu}\spacef,\label{eq:cefesScriDerF}\\
						&\ct{d}{_{\alpha\beta\gamma}^\mu}\ct{N}{_\mu}+\cd{_{[\alpha}}\prn{\ct{S}{_{\beta]\gamma}}}\eqs 0\spacef,\label{eq:cefesScriDerSchouten}\\
						&\ct{y}{_{\alpha\beta\gamma}}+\cd{_{\mu}}\ct{d}{_{\alpha\beta\gamma}^\mu}\eqs 0\spacef,\label{eq:cefesScriDerWeyl}\\
						&\ct{R}{_{\alpha\beta\gamma\delta}} \eqs\ct{g}{_{\alpha[\gamma}}\ct{S}{_{\delta]\beta}}-\ct{g}{_{\beta[\gamma}}\ct{S}{_{\delta]\alpha}}\spacef.	\label{eq:cefesScriRiemann}	
						\end{align}
\\
The idea now is to study the relation between the news $ \ctcn{N}{_{AB}} $ of \cref{thm:news-N} defined on a null hypersurface $ \prn{\N,\ms{_{ab}}} $  in conformal space-time $ \prn{M,\ct{g}{_{\alpha\beta}}} $ and the space-time fields entering in \cref{eq:cefesDerF,eq:cefesDerN,eq:cefesDerSchouten,eq:cefesDerWeyl,eq:cefesNormN,eq:cefesRiemann} . This is the kind of analysis that for $ \Lambda=0 $ leads to the relation between the covariant definition of the news \cite{Geroch1977} and the radiative components of the rescaled Weyl tensor. Here the treatment is more general, as any $ \Lambda $ is considered, and the particular case of null infinity with $ \Lambda=0 $ will be considered in any conformal gauge ---contrary to the common choice of a `divergence-free' conformal frame in other covariant approaches \cite{Geroch1977,Ashtekar2014,Fernandez-Alvarez_Senovilla2020a} and, also, to other more restricted versions such as the Bondi gauge of coordinate-based methods or in Newman-Penrose formalism. Also, for a spin-coefficient approach at null infinity in any conformal frame, see \cite{Frauendiener2021}. When $ \ct{N}{_{\alpha}} $ is null and normal to $ \prn{\ms{_{ab}},\N} $, the following components of the rescaled Weyl tensor will be considered:
	\begin{align}
		\ctn{D}{^{ab}}& \defeq\ct{\omega}{_{\alpha}^{a}}\ct{\omega}{_{\beta}^{b}}\ct{N}{^{\mu}}\ct{N}{^{\nu}}\ct{d}{^{\alpha}_{\mu}^{\beta}_{\nu}} ,\label{eq:D-null}\\
		\ctn{C}{^{ab}}&\defeq\ct{\omega}{_{\alpha}^{a}}\ct{\omega}{_{\beta}^{b}}\ct{N}{^{\mu}}\ct{N}{^{\nu}}\ctr{^*}{d}{^{\alpha}_{\mu}^{\beta}_{\nu}}\label{eq:C-null}\ .
	\end{align}
These tensors are not completely independent, as both contain the information determining $ \phi_{3,4} $, while $ \ctn{D}{^{ab}} $ contains $ \Re\prn{\phi_{2}} $ and $ \ctn{C}{^{ab}} $, $ \Im\prn{\phi_{2} }$ ---for a detailed description, see \cite{Fernandez-AlvarezSenovilla2022b}. Hence, these tensors are related to the gravitational `crossing' radiation with respect to $ \N $ and to the Coulomb potential of the gravitational field. 
	\subsection{Adapted null hypersurfaces in conformal space-time}
	The study of \cref{eq:cefesDerF,eq:cefesDerN,eq:cefesDerSchouten,eq:cefesDerWeyl,eq:cefesNormN,eq:cefesRiemann} on  null hypersurfaces in the bulk can be simplifiedby introducing the following notion
		\begin{deff}[Adapted null hypersurface]\label{def:adapted-hyp}
			A null hypersurface $ \prn{\N,\ms{_{ab}}} $  is said to be adapted in any of the following cases:
			\begin{enumerate}
			\item The conformal factor $ \Omega =$ constant$ \neq 0 $  and $ \cd{_{\alpha}}\Omega $ does not vanish at points in $ M\setminus\scri $ belonging to $ \N $\ .\label{it:adapted1}
			\item The conformal factor $\Omega =0$ and $ \cd{_{\alpha}}\Omega $ does not vanish at points of $ \N $. That is, $ \N\equiv\scri $ with $ \Lambda=0 $.\label{it:adapted2}
			\end{enumerate}
		\end{deff}
	Observe that  $ \prn{\N,\ms{_{ab}}} $ being adapted according to case \ref{it:adapted1} is not a conformal-gauge independent statement and that, at the same time, given a $ \prn{\N,\ms{_{ab}}} $ embedded in $ \prn{M,\ct{g}{_{\alpha\beta}}} $ it is always possible to find a \emph{family} of conformal frames where $ \prn{\N,\ms{_{ab}}} $ is adapted. To see this, let $ \Omega\evalat{\N}=\Psi $ be a function of the space-time coordinates. Then, \emph{considering only points of $ \N $ in $ M\setminus\scri $}, any $ \omega\Omega $ with $ \omega\evalat{ \N/\prn{\N\cap\scri} }=A\Psi^{-1}$, with $ A $ a positive constant there does the job. Of course, $ \scri $ with $ \Lambda=0 $ is always adapted (case \ref{it:adapted2}), as by construction $ \Omega=0 $ there \emph{independently of the gauge}. In that sense, $ \scri $ is special and will be treated separately. Also, the intersection between an adapted $ \N $ with $ \Omega= $constant $ \neq 0 $ and $ \scri $ will not be considered, as $ 1/\Psi $ blows up at the intersecting points and the construction cannot be extended up to the conformal boundary. Those points will be implicitly excluded in the rest of the analysis. Of course, one could do the analysis and apply the results of \cref{sec:news-null-hyp} on non-adapted hypersurfaces too, but this goes beyond the scope of the present work.\\
	
	From now on, consider an adapted null hypersurface $ \prn{\N,\ms{_{ab}}}  $. Then, $ \ct{N}{_{\alpha}} $ is normal to $ \N $ and null there, and $ \ct{N}{^a} $ is tangent to the generators. For convenience, define $ \ct{k}{_{\alpha}}\defeq \ct{N}{_{\alpha}} $. Results of \cref{sec:news-null-hyp} now apply, and the aim is to find equations relating \eqref{eq:D-null}  and \eqref{eq:C-null} with the $ \ctcn{N}{_{AB}} $ of \cref{eq:News} present on $ \prn{\N,\ms{_{ab}}} $. One can project \cref{eq:cefesDerN} to find
		\begin{equation}
			\ct{K}{_{ab}} = -\Omega\cts{S}{_{ab}}+ \cs{f}\cts{g}{_{ab}}+ \frac{1}{2}\Omega^2\varkappa\ctt{{\overline{T}}}{_{ab}}\label{eq:K-Schouten}
		\end{equation}
	and
		\begin{equation}\label{eq:kappa-N}
			\ctcn{\kappa}{_{AB}}=-\Omega\ctcn{S}{_{AB}}+f\mcn{_{AB}}+\frac{1}{2}\Omega^2\varkappa\cttcn{T}{_{AB}}\ ,
		\end{equation}
	which can be split into the shear and the expansion
		\begin{align}
			\ctcn{\nu}{_{AB}}&=-\Omega\prn{\ctcn{S}{_{AB}}-\frac{1}{2}\mcn{_{AB}}\ctcn{S}{^{M}_{M}}}+\frac{1}{2}\Omega^2\varkappa\prn{\ctcn{T}{_{AB}}-\frac{1}{2}\mcn{_{AB}}\ctcn{T}{^{M}_{M}}}\ ,\label{eq:shear-N}\\
			\ctcn{\kappa}{}&=-\Omega\ctcn{S}{^{M}_{M}}+2f+\frac{1}{2}\Omega^2\varkappa\prn{\ctcn{T}{^{M}_{M}}-\frac{1}{2}T}\ .\label{eq:expansion-N}
		\end{align}	
	Observe that from \cref{eq:expansion-N,eq:schouten-trace} one can write
		\begin{equation}\label{eq:phi2-R-shear}
			-\Omega\ct{C}{_{AB}^{AB}}+\Omega\ctcn{R}{}+\Omega\ctcn{\theta}{}\ctcn{\kappa}{}-2\Omega\ctcn{\nu}{_{AC}}\ctcn{\theta}{^{AC}}+2\ctcn{\kappa}{}-4f-\Omega^2\varkappa\prn{\ctcn{T}{^{M}_{M}}-\frac{1}{2}T}=0\ .
		\end{equation}	
	Apart from this, one can contract \cref{eq:cefesDerN}	with $ \ct{\ell}{^{\beta}}\ct{e}{^{\alpha}_{a}} $,
		\begin{equation}\label{eq:varphi}
			\ct{\varphi}{_{a}}=\frac{1}{2}\Omega\ct{e}{^{\alpha}_{a}}\ct{\ell}{^{\beta}}\ct{S}{_{\alpha\beta}}-f\ct{\ell}{_{a}}-\frac{1}{2}\Omega^2\varkappa\ctt{T}{_{\alpha\beta}}\ct{e}{^{\alpha}_{a}}\ct{\ell}{^{\beta}}
		\end{equation}
	 and then with $ \ct{N}{^a} $ to get
		\begin{equation}\label{eq:nu}
			\nu=\frac{1}{2}\Omega\ct{S}{_{\mu\nu}}\ct{\ell}{^\mu}\ct{N}{^{\nu}}+f-\frac{1}{2}\Omega^2\varkappa\ctt{T}{_{\mu\nu}}\ct{k}{^{\mu}}\ct{\ell}{^{\nu}}\ ,
		\end{equation}
	Also, using $ \ct{N}{_{\alpha}}\ct{N}{^{\alpha}}\eqn 0 $, \cref{eq:cefesNormN} reads
		\begin{equation}\label{eq:f-Lambda}
			2\Omega f \eqn \frac{\Lambda}{3}-\frac{\Omega^3}{12}\varkappa T\ .
		\end{equation}
	It can be shown (see appendix of \cite{Fernandez-AlvarezSenovilla2022b}) that to compute the crossing radiative components of the rescaled Weyl tensor, it is enough to calculate
		\begin{equation}\label{eq:C-null-misxed}
			\ctn{C}{^{a}_{c}}\defeq \ms{_{cd}}\ctn{C}{^{ad}}\ .
		\end{equation}
	Alternatively, the same information is contained in the componentes
		\begin{align}
			\ctncn{C}{^{A}}\defeq \ct{\ell}{_b}\ctcn{W}{_{a}^{A}}\ctn{C}{^{ab}}\ ,\\
			\ctncn{D}{_{AB}}\defeq \ctcn{E}{^{a}_{A}}\ctcn{E}{^{b}_{B}}\ctn{D}{_{ab}}\ .
		\end{align}
	Observe that these two objects are two-dimensional and the metric $ \mcn{^{AB}} $ can be used to rise and lower indices, which makes them more suitable. Taking the corresponding contractions and pullbacks in \cref{eq:cefesDerSchouten},
		\begin{align}
			\ctncn{C}{_{A}}&=\ctcn{\epsilon}{^{SR}}\prn{\cdcn{_{R}}\ctcn{S}{_{SA}}-\ctcn{\kappa}{_{AR}}\ctlcn{S}{_{S}}-\ctcn{\theta}{_{AR}}\ctncn{S}{_{S}}}+\Omega\ctcnr{^*}{y}{_{A}}\ ,\\	\ctncn{D}{_{AB}}&=\lied_{\vec{N}}\ctcn{S}{_{AB}}-\ctcn{\kappa}{^E_{B}}\prn{\ctcn{S}{_{AE}}-\frac{1}{2}\ct{S}{_{\mu\nu}}\ct{N}{^{\mu}}\ct{\ell}{^{\nu}}\mcn{_{EA}}}-2\ctncn{S}{_{(B}}\prn{\ctcn{E}{^{m}_{A)}}\ct{N}{^{p}}\ct{H}{_{pm}}-\ctcn{\varphi}{_{A)}}}+\ctncn{S}{_{B}}\ctcn{\varphi}{_{A}}\nonumber\\
			&-\cdcn{_{A}}\ctncn{S}{_{B}}+\cts{S}{_{ab}}\ct{N}{^{a}}\ct{N}{^{b}}\ctcn{\theta}{_{AB}}+\Omega\ctcn{y}{_{AB}}  \ ,
		\end{align}
	where one has to use \cref{eq:liedE} written as
		\begin{equation}
			\lied_{\vec{N}}\ctcn{E}{^{a}_{A}}=\prn{\ctcn{E}{^{m}_{A}}\ct{N}{^{p}}\ct{H}{_{pm}}-\ctcn{\varphi}{_{A}}}\ct{k}{^{a}},
		\end{equation}
	and there are matter contributions entering through
	\begin{align}
	\ctcnr{^*}{y}{_{A}}\defeq& \frac{1}{2}\ctcn{E}{^{\alpha}_{A}}\ct{N}{^\rho}\ct{\ell}{_{\sigma}}\ct{\eta}{^{\mu\nu}_{\alpha\rho}}\ct{y}{_{\mu\nu}^\sigma}\ ,\\
	\ctcn{y}{_{AB}}\defeq& \ctcn{E}{^{\alpha}_{A}}\ctcn{E}{^{\beta}_{B}}\ct{N}{^{\mu}}\ct{y}{_{\alpha\mu\beta}}\ 
	\end{align}
	and \cref{eq:cefesDerF} that vanish for $ \ct{T}{_{\alpha\beta}}\eqn 0 $ ---these are absent in the special case $ \Omega=0 $ that is studied in \cref{sec:infinity}. In this form, \cref{eq:U-definition} can be introduced, using decomposition \cref{eq:News}, and also one can use \cref{eq:Sell,eq:SN,eq:cefesDerF}, leading to
		\begin{align}
			\ctncn{C}{_{A}}&=\ctcn{\epsilon}{^{SR}}\Biggl[\cdcn{_{R}}\ctcn{N}{_{SA}}-\cdcn{_{R}}\ctcn{L}{_{SA}}-\ctcn{\kappa}{_{AR}}\prn{\cdcn{_{M}}\ctcn{\sigma}{_{S}^{M}}-\frac{1}{2}\cdcn{_{S}}\ctcn{\theta}{}-\frac{1}{4}\ctcn{\theta}{}\ctcn{\varphi}{_{S}}+\frac{1}{2}\ctcn{\varphi}{^{M}}\ctcn{\sigma}{_{SM}}+\ctlcn{d}{_{S}}}\\
			&-\ctcn{\theta}{_{AR}}\prn{\cdcn{_{M}}\ctcn{\nu}{_{S}^{M}}-\frac{1}{2}\cdcn{_{S}}\ctcn{\kappa}{}+\ctncn{d}{_{S}}}\Biggr]+\Omega\ctcnr{^*}{y}{_{A}}\ \label{eq:C-news},\\
			\ctncn{D}{_{AB}}&=\lied_{\vec{N}}\ctcn{N}{_{AB}}+\lied_{\vec{N}}\ctcn{\rho}{_{AB}}-                  \lied_{\vec{N}}\ctcn{L}{_{AB}}-\ctcn{\kappa}{^{E}_{A}}\prn{-\frac{1}{2}\ct{S}{_{\mu\nu}}\ct{N}{^{\mu}}\ct{\ell}{^{\nu}}\mcn{_{EB}}+\ctcn{N}{_{EB}}+\ctcn{\rho}{_{EB}}-\ctcn{L}{_{EB}}}
			\nonumber\\
			&-2\ctncn{S}{_{(B}}\prn{\ctcn{E}{^{m}_{A)}}\ct{N}{^{p}}\ct{H}{_{pm}}-\ctcn{\varphi}{_{A)}}}+\ctncn{S}{_{B}}\ctcn{\varphi}{_{A}}-\cdcn{_{A}}\ctncn{S}{_{C}}+\cts{S}{_{ab}}\ct{N}{^{a}}\ct{N}{^{b}}\ctcn{\theta}{_{AC}}
			+\Omega\ctcn{y}{_{AB}}\ ,\label{eq:D-news}
		\end{align}
	where
		\begin{equation}
			\ctcn{L}{_{AB}}\defeq-\frac{1}{2}\prn{\ctcn{\theta}{}\ctcn{\nu}{_{AB}}+\ctcn{\kappa}{}\ctcn{\sigma}{_{AB}}}+\frac{1}{2}\mcn{_{AB}}\prn{\ctcn{\nu}{_{MN}}\ctcn{\sigma}{^{MN}}-\frac{1}{2}\ctcn{\kappa}{}\ctcn{\theta}{}+\frac{1}{2}\ctcn{C}{_{MN}^{MN}}}\ .
		\end{equation}
	Using \cref{eq:N-lied-sigma}, one can eliminate $ \ctcn{N}{_{AB}} $ and $ \ctcn{\rho}{_{AB}} $ from the last term on the right-hand side of \cref{eq:D-news},
		\begin{align}
				\ctncn{D}{_{AB}}&=\lied_{\vec{N}}\ctcn{N}{_{AB}}+\lied_{\vec{N}}\ctcn{\rho}{_{AB}}-  \lied_{\vec{N}}\ctcn{L}{_{AB}}\nonumber\\
				&-\ctcn{\kappa}{^{E}_{A}}\Bigl[-\frac{1}{2}\ct{S}{_{\mu\nu}}\ct{N}{^{\mu}}\ct{\ell}{^{\nu}}\mcn{_{EB}}+\lied_{\vec{N}}\ctcn{\sigma}{_{EB}}+\frac{1}{2}\ctcn{\nu}{_{EB}}\ctcn{\theta}{}-\frac{1}{2}\ctcn{\sigma}{_{EB}}\ctcn{\kappa}{}-2\ctcn{\nu}{_{C(E}}\ctcn{\sigma}{_{B)}^{C}}-\frac{1}{2}\ctcn{L}{^{M}_{M}}\mcn{_{EB}}\nonumber\\
				&+\ct{k}{^d}\ct{\varphi}{_{d}}\ctcn{\sigma}{_{EB}}-\ctcn{\varphi}{_{E}}\ctcn{\varphi}{_{B}}-\cdcn{_{(B}}\ctcn{\varphi}{_{E)}}+2\ctcn{\varphi}{_{(E}}\cdcn{_{B)}}\ln\dot{F}-\frac{1}{2}\mcn{_{EB}}\prn{2\ctcn{\varphi}{^{M}}\cdcn{_{M}}\ln\dot{F}-\ctcn{\varphi}{_{C}}\ctcn{\varphi}{^{C}}-\cdcn{_{C}}\ctcn{\varphi}{^{C}}}\nonumber\\
				&+\ct{F}{_{EB}}+\frac{1}{2}\mcn{_{EB}}\ctcn{K}{}\Bigr]-2\ctncn{S}{_{(B}}\prn{\ctcn{E}{^{m}_{A)}}\ct{N}{^{p}}\ct{H}{_{pm}}-\ctcn{\varphi}{_{A)}}}+\ctncn{S}{_{B}}\ctcn{\varphi}{_{A}}-\cdcn{_{A}}\ctncn{S}{_{C}}+\cts{S}{_{ab}}\ct{N}{^{a}}\ct{N}{^{b}}\ctcn{\theta}{_{AC}}
							+\Omega\ctcn{y}{_{AB}}\ .
		\end{align}
	Alternatively to \cref{eq:N-lied-sigma}, one can write an expression for $ \ctcn{N}{_{AB}} $ using \cref{eq:shear-N,eq:U-definition,eq:News},
		\begin{equation}\label{eq:news-asymptotic-shear}
			\ctcn{N}{_{AB}}= -\prn{\ctcn{\rho}{_{AB}}-\frac{1}{2}\mcn{_{AB}}\ctcn{K}{}}-\frac{1}{2}\prn{\ctcn{\theta}{}\ctcn{\nu}{_{AB}}+\ctcn{\kappa}{}\ctcn{\sigma}{_{AB}}}-\frac{1}{\Omega}\ctcn{\nu}{_{AB}}+\frac{1}{2}\Omega\varkappa\prn{\ctcn{T}{_{AB}}-\frac{1}{2}\mcn{_{AB}}\ctcn{T}{^{M}_{M}}}\ ,
		\end{equation}
	but recall that to obtain \cref{eq:shear-N} the CEFE have been used, whereas \cref{eq:N-lied-sigma} is independent of any field equations. Observe that the term with a negative power of $ \Omega $ is regular at $ \Omega=0 $ as from \cref{eq:shear-N} the shear $ \ctcn{\nu}{_{AB}} $ vanishes at infinity. Importantly, one sees that for any adapted null hypersurface, the news carries a contribution from matter fields ---the last two terms of \cref{eq:news-asymptotic-shear}--- that vanish at $ \Omega=0 $.\\
	
	 All in all, \Cref{eq:C-news,eq:D-news} show that hypersurface-crossing gravitational radiation is not determined, in general, by a news tensor $ \ctcn{N}{_{AB}} $ alone when the null hypersurface $ \N $ is embedded in the bulk of the conformal space-time. Instead, this crossing radiation is an interplay between the news $ \ctcn{N}{_{AB}} $, the shears, $ \ctcn{\sigma}{_{AB}} $ and $ \ctcn{\nu}{_{AB}} $, and the expansions, $ \ctcn{\kappa}{} $ and $ \ctcn{\theta}{} $, of the null normals $ \ct{N}{^{\alpha}} $ and $ \ct{\ell}{^{\alpha}} $ to the leaves, the choice of foliation $ F $ and, remarkably,  the tangential radiative components $ {\ctl{d}{_{A}}} $ and the Coulomb part $ \ct{C}{_{AB}^{AB}} $ of the Weyl tensor\footnote{In Newman-Penrose notation, $ \ct{\Psi}{_{1}} $ and $ \ct{\Psi}{_{2}} $, respectively.} $ \ct{C}{_{\alpha\beta\gamma}^{\delta}} $. A lot of simplifications occur when $ \Omega=0 $, i.e., when $ \N=\scri $, and the meaning of $ \ctcn{N}{_{AB}} $ is clearly unveiled then. Also, there are restricted bulk cases where the equations get simpler. These aspects are studied in \cref{sec:infinity,sec:special}, respectively. 
	\subsection{The case of infinity with vanishing cosmological constant}\label{sec:infinity}
	The equations computed so far are valid at infinity $ \scri $ ---$ \Omega=0 $--- only when $ \Lambda=0 $, which, as it is well known, makes the character of $ \scri $ lightlike ---see \cref{eq:f-Lambda}. This is the case considered in this section. Also, at $ \scri $ a number of simplifications take place:
		\begin{enumerate}
		\item The Weyl tensor $ \ct{C}{_{\alpha\beta\gamma}^{\delta}} $ vanishes\footnote{It might be the case, under the assumption of low differentiability of the metric at $ \scri $, that some components of the Weyl tensor survive. This possibility is not considered in this work.} \cite{NewmanRPAC1989}.
		\item $ \prn{\scri,\ms{_{ab}}} $ is umbilical ($ \ctcn{\nu}{_{AB}}=0 $) but, in general, it is expanding $ \ctcn{\kappa}{}\neq0 $ and the leaves are not isometric.
		\item In fact, the expansion of the generators $ \ctcn{k}{} $, Friedrich scalar $ f $ and the acceleration scalar $ \nu $ of \cref{eq:k-geodesic-nu} coincide
			\begin{equation}
				2\nu=2f=\ctcn{\kappa}{}\ .
			\end{equation} 
		This can be seen from \cref{eq:nu,eq:shear-N,eq:expansion-N} evaluated at $ \Omega=0 $. 
		\item  As a consequence of \cref{eq:varphi}, the projection $ \ctcn{\varphi}{_{A}} $  of $ \ct{\varphi}{_{a}} $ to the leaves vanishes.
		\end{enumerate}
		 If one takes the limit $ \Omega=0 $ in \cref{eq:news-asymptotic-shear} one gets an  expression for the news tensor at infinity with $ \Lambda=0 $,
				\begin{equation}\label{eq:news-asymptotic-shear-limit}
					\ctcn{N}{_{AB}}\eqs -\prn{\ctcn{\rho}{_{AB}}-\frac{1}{2}\mcn{_{AB}}\ctcn{K}{}}-\frac{1}{2}\ctcn{\kappa}{}\ctcn{\sigma}{_{AB}}-\lim_{\Omega\rightarrow 0}\prn{\frac{1}{\Omega}\ctcn{\nu}{_{AB}}}\ .
				\end{equation}
			Observe that this formula for the case of round isometric cuts at $ \scri $ reduces to
				\begin{equation}
					\ctcn{N}{_{AB}}\eqs-\lim_{\Omega\rightarrow 0}\prn{\frac{1}{\Omega}\ctcn{\nu}{_{AB}}}\ ,
				\end{equation}
			but this is achieved only at the cost of fixing the conformal gauge. 
			\Cref{eq:news-asymptotic-shear-limit} is related to other expressions in the literature that express the news as the asymptotic shear of $ \cd{_{\alpha}}\Omega $ ---see \cite{Madler2016}.   Alternatively, one can get a different general expression  by 	evaluating  \cref{eq:N-lied-sigma} at $ \scri $ --recall that $ \ctcn{F}{_{AB}} $ depends on the choice of foliation and it is defined in \cref{eq:combination-F}--
					\begin{equation}
						\ctcn{N}{_{AB}}\eqs\lied_{\vec{N}}\ctcn{\sigma}{_{AB}}-\frac{1}{2}\ctcn{\sigma}{_{AB}}\ctcn{\kappa}{}- \ctcn{\rho}{_{AB}}+\frac{1}{2}\mcn{_{AB}}\ctcn{K}{}+\ct{F}{_{AB}}\ .\label{eq:N-lied-sigma-scri}
					\end{equation}
		For a gauge with round isometric cuts and choosing a \emph{canonically adapted foliation} \cite{Fernandez-AlvarezSenovilla2022a}, it gets simplified, 
				\begin{equation}\label{eq:N-lied-sigma-scri-gauge}
					\ctcn{N}{_{AB}}\eqs\lied_{\vec{N}}\ctcn{\sigma}{_{AB}}\ .
				\end{equation}
		\Cref{eq:N-lied-sigma-scri-gauge} is the form used in the Newman-Penrose approach \cite{Newman2009,Adamo2009}. \Cref{eq:news-asymptotic-shear-limit,eq:N-lied-sigma-scri}, remarkably, do not assume any particular conformal gauge nor geometry of the cuts.
		
	Next, using the Bianchi identity for the space-time curvature,
		\begin{equation}
			\cd{_{[\alpha}}\ct{R}{_{\beta\gamma]\delta}^{\mu}}=0\ ,
		\end{equation}
	it follows taking traces that
		\begin{equation}\label{eq:bianchi-id}
			\cd{_{\mu}}\ct{S}{_{\beta}^{\mu}}-\cd{_{\beta}}\ct{S}{_{\rho}^{\rho}}=0\ .
		\end{equation}
	A long calculation, decomposing this last equation with respect to the projector to the leaves \eqref{eq:projector-leaves}, shows that 
		\begin{equation}
			\lied_{\vec{N}}\ctcn{S}{^{M}_{M}}-2\ctcn{\kappa}{}\lied_{\vec{\ell}}f+f\ctcn{S}{^{M}_{M}}+\frac{1}{2}\cdcn{_{M}}\cdcn{^{M}}\ctcn{\kappa}{}-\frac{1}{2}\ctcn{\theta}{}\lied_{\vec{N}}\ctcn{\kappa}{}-\frac{1}{2}\cdcn{_{M}}\prn{\ctcn{\kappa}{}}\cdcn{^{M}}\prn{\ln \dot{F}}\eqs 0\ .
		\end{equation}
	Using \cref{eq:schouten-trace} evaluated at $ \scri $, this equation gives the evolution along $ \ct{N}{^{a}} $ of the curvature $ \ctcn{R}{} $ ---which recall that  gives the two-dimensional scalar curvature at each leaf:
		\begin{equation}\label{eq:lied-R-scri}
			\lied_{\vec{N}}\ctcn{R}{}+\frac{\ctcn{\kappa}{}}{2}\ctcn{R}{}+\ctcn{\kappa}{}\prn{\lied_{\vec{N}}\ctcn{\theta}{}+\frac{\ctcn{\kappa}{}}{2}\ctcn{\theta}{}}-2\ctcn{\kappa}{}\lied_{\vec{\ell}}f+\cdcn{_{M}}\cdcn{^{M}}\ctcn{\kappa}{}-2\cdcn{_{M}}\prn{\ctcn{\kappa}{}}\cdcn{^{M}}\prn{\ln \dot{F}}\eqs 0\ .
		\end{equation}
	Observe that choosing a gauge such that $ \ctcn{\kappa}{}\eqs 0 $ (and, therefore $ f=0=\nu $ too), one gets $ \lied_{\vec{N}}\ctcn{R}{}=0 $, meaning that all leaves are isometric. This kind of partial gauge fixing (any conformal change with $ \vec{N}\prn{\omega}=0 $ is still allowed) is often called \emph{divergence-free} gauge (motivated by \cref{eq:cefesDerN}). In this section, \emph{the gauge is not fixed in any way}. The aim is to show that the news tensor given in \cref{thm:news-N} ---or, equivalently, the combination involving the asymptotic shear in \cref{eq:news-asymptotic-shear}--- determines in full generality the gravitational radiation arriving at infinity. For such goal, first one has to evaluate \cref{eq:C-news,eq:D-news} at $ \scri $. Taking into account all the simplifications listed at the beginning of this section, \cref{eq:D-news} reads
		\begin{align}\label{eq:D-scri-1}
			\ctncn{D}{_{AC}}\eqs& \lied_{\vec{N}}\ctcn{V}{_{AC}}+\lied_{\vec{N}}\ctcn{\rho}{_{AC}}+\lied_{\vec{N}}\prn{\frac{1}{2}\ctcn{\kappa}{}\ctcn{\sigma}{_{AC}}+\frac{1}{4}\mcn{_{AC}}\ctcn{\theta}{}\ctcn{\kappa}{}}-\ctcn{\kappa}{}\lied_{\vec{\ell}}\prn{f}\mcn{_{AC}}-\frac{1}{2}\ctcn{\kappa}{}\prn{\ctcn{V}{_{AC}}+\ctcn{\rho}{_{AC}}}\nonumber\\
			&+\frac{1}{8}\ctcn{\kappa}{}^2\ctcn{\theta}{}\mcn{_{AC}}-\cdcn{_{(C}}\ctcn{\kappa}{}\cdcn{_{A)}}\ln\dot{F}+\frac{1}{2}\cdcn{_{A}}\cdcn{_{C}}\ctcn{\kappa}{}-\frac{1}{2}\ctcn{\theta}{_{AC}}\lied_{\vec{N}}\ctcn{\kappa}{}\ .
		\end{align}
	Taking the trace of this equations and using \cref{eq:lied-R-scri} it follows that $ \ctcn{D}{^{M}_{M}}=0 $ ---indeed, this serves as a consistency check of \cref{eq:D-scri-1}, as $ \ctcn{D}{^M_{M}} $ has to vanish by the traceless property of the Weyl tensor. Then, 
		\begin{align}
			\ctncn{D}{_{AC}}&\eqs \lied_{\vec{N}}\ctcn{N}{_{AC}}+\lied_{\vec{N}}\prn{\ctcn{\rho}{_{AC}}-\frac{1}{2}\ctcn{K}{}\mcn{_{AC}}}+\frac{1}{2}\kappa\lied_{\vec{N}}\prn{\ctcn{\sigma}{_{AC}}}-\frac{1}{2}\ctcn{\kappa}{}\ctcn{N}{_{AC}}\nonumber\\
			&-\frac{1}{2}\ctcn{\kappa}{}\prn{\ctcn{\rho}{_{AC}}-\frac{1}{2}\ctcn{K}{}\mcn{_{AC}}}-\brkt{\cdcn{_{(A}}\prn{\ln\dot{F}}\cdcn{_{C)}}\ctcn{\kappa}{}-\frac{1}{2}\mcn{_{AC}}\cdcn{_{M}}\prn{\ln\dot{F}}\cdcn{^{M}}\ctcn{\kappa}{}}\nonumber\\
			&+\frac{1}{2}\prn{\cdcn{_{A}}\cdcn{_{C}}\ctcn{\kappa}{}-\frac{1}{2}\mcn{_{AC}}\cdcn{_{M}}\cdcn{^{M}}\ctcn{\kappa}{}}\ .\label{eq:D-scri-2} 
		\end{align}

	and  \cref{eq:D-scri-2} can be written as
		\begin{align}
			\ctncn{D}{_{AC}}&\eqs \lied_{\vec{N}}\ctcn{N}{_{AC}}+\lied_{\vec{_{N}}}\prn{\ctcn{\rho}{_{AC}}-\frac{1}{2}\mcn{_{AC}}\ctcn{K}{}}+\frac{1}{2}\ctcn{\kappa}{}\ctcn{F}{_{AC}}\nonumber\\
			&-\brkt{\cdcn{_{(A}}\prn{\ln\dot{F}}\cdcn{_{C)}}\ctcn{\kappa}{}-\frac{1}{2}\mcn{_{AC}}\cdcn{_{M}}\prn{\ln\dot{F}}\cdcn{^{M}}\ctcn{\kappa}{}}+\frac{1}{2}\prn{\cdcn{_{A}}\cdcn{_{C}}\ctcn{\kappa}{}-\frac{1}{2}\mcn{_{AC}}\cdcn{_{M}}\cdcn{^{M}}\ctcn{\kappa}{}}
		\end{align}
	Notice that due to the conformal invariance of the Weyl tensor, $ \ctcn{D}{_{AC}} $ rescales under conformal gauge transformations as $ \ctcng{D}{_{AC}}=\omega^{-1}\ctcn{D}{_{AC}} $. Also, using that $ \ctcn{N}{_{AC}} $ is invariant under conformal rescalings, $ \lied_{\vec{N}}\ct{N}{_{AC}} $ transforms in that same way. This implies that the traceless part of $ \ctcn{\rho}{_{AC}} $ must satisfy the gauge invariant equation
		\begin{align}
			\ct{P}{_{AC}}\eqs&\lied_{\vec{_{N}}}\prn{\ctcn{\rho}{_{AC}}-\frac{1}{2}\mcn{_{AC}}\ctcn{K}{}}+\frac{1}{2}\ctcn{\kappa}{}\ctcn{F}{_{AC}}
			-\brkt{\cdcn{_{(A}}\prn{\ln\dot{F}}\cdcn{_{C)}}\ctcn{\kappa}{}-\frac{1}{2}\mcn{_{AC}}\cdcn{_{M}}\prn{\ln\dot{F}}\cdcn{^{M}}\ctcn{\kappa}{}}\nonumber\\
			&+\frac{1}{2}\prn{\cdcn{_{A}}\cdcn{_{C}}\ctcn{\kappa}{}-\frac{1}{2}\mcn{_{AC}}\cdcn{_{M}}\cdcn{^{M}}\ctcn{\kappa}{}}\label{eq:lied-rho-1}
		\end{align}
	for some unknown tensor $ \ct{P}{_{AC}} $ transforming as $ \ctg{P}{_{AC}}=\omega^{-1}\ct{P}{_{AC}} $. To find this $ \ct{P}{_{AC}} $, it is  enough to evaluate \cref{eq:lied-rho-1} in one gauge. In particular, choosing round cuts and vanishing $ \ctcn{\kappa}{} $, 
		\begin{equation}
			\ct{P}{_{AC}}\eqs0\ .
		\end{equation}
	That is, the tensor $ \ctcn{\rho}{_{AC}} $, whose existence and uniqueness is warranted by \cref{thm:rho-tensor-foliation}, also satisfies 
		\begin{align}
				0\eqs&\lied_{\vec{_{N}}}\prn{\ctcn{\rho}{_{AC}}-\frac{1}{2}\mcn{_{AC}}\ctcn{K}{}}+\frac{1}{2}\ctcn{\kappa}{}\ctcn{F}{_{AC}}
				-\brkt{\cdcn{_{(A}}\prn{\ln\dot{F}}\cdcn{_{C)}}\ctcn{\kappa}{}-\frac{1}{2}\mcn{_{AC}}\cdcn{_{M}}\prn{\ln\dot{F}}\cdcn{^{M}}\ctcn{\kappa}{}}\nonumber\\
				&+\frac{1}{2}\prn{\cdcn{_{A}}\cdcn{_{C}}\ctcn{\kappa}{}-\frac{1}{2}\mcn{_{AC}}\cdcn{_{M}}\cdcn{^{M}}\ctcn{\kappa}{}}\label{eq:lied-rho-2}\ .
		\end{align}
	Observe that the evolution along $ \ct{N}{^{a}} $ of the trace, $ \ctcn{\rho}{^{M}_{M}}=\ctcn{K}{}$, is just given by \cref{eq:lied-R-scri}\footnote{In particular, for a divergence free gauge, $ \ctcn{\kappa}{}=0 $ and one recovers the classic equation by Geroch, $ \lied_{\vec{N}}\ct{\rho}{_{AB}}=0 $}. \Cref{eq:lied-rho-2} implies then that
		\begin{equation}\label{eq:D-news-scri}
			\ctncn{D}{_{AB}}\eqs\lied_{\vec{N}}\ctcn{N}{_{AB}}\ .
		\end{equation}
	One has now to study the other radiative component, \cref{eq:C-news}. Again, applying  all the simplifications available at $ \scri $ in a general conformal gauge, and using the fact that $ \cdcn{_{[R}}\ctcn{\sigma}{_{S]A}}=\mcn{_{A[R}}\cdcn{_{|M|}}\ctcn{\sigma}{^{M}_{S]}} $ in 2 dimensions,
		\begin{equation}\label{eq:C-news-scri}
			\ctncn{C}{_{A}}\eqs -\ctcn{\epsilon}{^{RS}}\cdcn{_{R}}\ctcn{N}{_{SA}}\ .
		\end{equation}
	\Cref{eq:D-news-scri,eq:C-news-scri} show that there is no gravitational radiation arriving at $ \scri $ if and only if $ \ctcn{N}{_{AB}} $ vanishes, thus generalising to arbitrary gauge the classic result of \cite{Geroch1977}. Hence, it is clear that the news defined by \cref{eq:News} rules the gravitational radiation when the null hypersurface corresponds to $ \scri $ with $ \Lambda=0 $. In the bulk, things differ notably and some special cases are studied in the next section.
\subsection{Special cases in the bulk}\label{sec:special}
	In the next, particular cases are considered, all of them treated when there are no matter fields on $ \N $  , i.e., $ \ct{y}{_{\alpha\beta\gamma}}\eqn0\eqn\ct{T}{_{\alpha\beta}} $. Then, for adapted null hypersurfaces $ \N $  with $ \Omega\neq 0 $ ---case \ref{it:adapted1} of \cref{def:adapted-hyp}--- \cref{eq:f-Lambda} yields
		\begin{equation}\label{eq:simp-bulk}
			\cds{_{a}}f=0\ ,
		\end{equation}
	 independently of $ \Lambda $. 
	 Taking this into account and using \cref{eq:cefesDerSchouten,eq:K-Schouten}, \eqref{eq:C-null-misxed} can be expressed as
			\begin{equation}\label{eq:C-K}
				\ctn{C}{^{a}_{c}}=\frac{1}{\Omega}\cts{\epsilon}{^{ars}}\brkt{\cds{_{r}}\ct{K}{_{sc}}+\ctcn{\varphi}{_{s}}\ct{K}{_{rc}}}\ .
			\end{equation}
	Also, contracting with $ \ct{N}{^\alpha} $ the two free indices of \cref{eq:cefesDerN}, and recalling that $ \ct{T}{_{\alpha\beta}}=0 $, one gets $ \ct{N}{^a}\ct{N}{^b}\cts{S}{_{ab}}=0 $, i.e., $ \ct{R}{_{\alpha\beta}}\ct{N}{^{\alpha}}\ct{N}{^{\beta}} =0$. Then, Raychaudhuri equation reads
			\begin{equation}\label{eq:raychaudhuri}
				\lied_{\vec{N}}\ctcn{\kappa}{}\eqn \ctcn{\kappa}{}\ctcn{\nu}{}-\ctcn{\nu}{_{ab}}\ctcn{\nu}{^{ab}}-\frac{1}{2}\ctcn{\kappa}{}^2\ .
			\end{equation}
	The purpose of these examples is not to study the role of news tensor in the bulk in full generality (i.e., by working with \cref{eq:C-news,eq:D-news}), but to give very simplified examples to see whether there is some determining relation between news and the presence of gravitational radiation in the bulk.
	\subsubsection*{Non-expanding and (weakly) isolated horizons}
		For an adapted null hypersurface to have the structure of a non-expanding horizon (NEH) ---for a review on this and other definitions of horizons, see \cite{Ashtekar2004,Booth2005}--- one has to require, additionally, that $ \ctcn{k}{}=0 $. Then,   Raychaudhuri equation \eqref{eq:raychaudhuri} readily implies $ \ctcn{\nu}{_{AB}}=0 $ , that is, $ \ct{K}{_{ab}}=0 $, and \cref{eq:C-K} gives
			\begin{equation}
				\ctn{C}{^{a}_{b}}=0\ .
			\end{equation}
		Hence, $ \ctncn{C}{_{A}}=0 $ and $ \ctncn{D}{_{AB}}=0 $ as well. This can be considered as a non-radiating case (with $ \phi_{3,4}=0 $). However, in general it does not correspond to a case with vanishing news $ \ctcn{N}{_{AB}} $, since \cref{eq:news-asymptotic-shear} reads
			\begin{equation}
				\ctcn{N}{_{AB}}= \frac{1}{2}\mcn{_{AB}}\ctcn{K}{}-\ctcn{\rho}{_{AB}}\ .
			\end{equation}
		By \cref{thm:rho-tensor}, one only gets $ \ctcn{N}{_{AB}}=0 $ when the metric on the leaves of the foliation is that of the round sphere ---which is an extra restriction. Weakly isolated horizons (WIH) and isolated horizons (IH) are more restricted versions of NEH. Thus, the previous analysis holds on these cases too.\\
		
		What this shows is that the role of $ \ctcn{N}{_{AB}} $ on adapted null hypersurfaces in the bulk that are NEH, WIH or IH is different from the case at $ \scri $ with $ \Lambda=0 $. On a NEH there is no radiation crossing the null hypersurface but the news vanishes if and only if the metric on the cuts is round. Contrary, at $ \scri $, the presence of outgoing radiation  is determined by $ \ctcn{N}{_{AB}} $ \emph{without any assumption on the geometry of the cuts} ---and both cases, i.e., presence and absence of gravitational waves are possible. 
	\subsubsection*{Special case of a null marginally trapped tube}
		Instead of asking for a NEH structure, one can consider an adapted null hypersurface foliated by marginally trapped surfaces by requiring $ \ctcn{\theta}{}=0 $ and $ \ctcn{\kappa}{}<0 $. Then, $ \N $ is a null marginally trapped tube (MTT) ---see \cite{Senovilla2012,Senovilla2023} for more details on this class of hypersurfaces and related definitions. The analysis of \cref{eq:C-news,eq:D-news} is still quite involved, and to get a case where the role of the news $ \ctcn{N}{_{AB}} $ has a clearer meaning one has to make further assumptions. One possibility is to consider a round metric on the leaves and $ \ctcn{\sigma}{_{AB}}=0 $. Observe that this restriction is very strong by itself, but the aim is to study what happens in the most simplified cases. Then, on $ \N $, \cref{eq:news-asymptotic-shear} reads
			\begin{equation}
				\ctcn{N}{_{AB}}=-\frac{1}{\Omega}\ctcn{\nu}{_{AB}}\ .
			\end{equation}
		Also, 	observe that taking the covariant derivative of \cref{eq:phi2-R-shear}, and using \cref{eq:simp-bulk}, 
					\begin{equation}
						\cdcn{_{A}}\ct{C}{_{PQ}^{PQ}}=\frac{2}{\Omega}\cdcn{_{A}}\ctcn{\kappa}{}\ ,
					\end{equation}
		which in general is non-vanishing. Therefore
			\begin{align}
				\ctncn{D}{_{AC}}&=\ctcn{E}{^{a}_{A}}\ctcn{E}{^{c}_{C}}\ct{k}{^{b}}\cds{_{b}}\prn{\ctcn{N}{_{ac}}}-\ctcn{N}{_{AC}}\prn{\nu-\ctcn{\kappa}{}}\ \ ,\\
				\ctncn{C}{_{A}}&=\ctcn{\epsilon}{^{SR}}\brkt{\cdcn{_{R}}\ctcn{N}{_{SA}}+\Omega\ctcn{N}{_{AR}}\ctlcn{d}{_{S}}-\frac{1}{2}\mcn{_{AR}}\ctcn{\kappa}{}\ctlcn{d}{_{S}}-\frac{1}{4}\mcn{_{AS}}\cdcn{_{R}}\ct{C}{_{PQ}^{
				PQ}}}\ .
			\end{align}
		It is clear that if $ \ctcn{N}{_{AB}}=0 $, then $ \ctncn{D}{_{AB}}=0 $. Still, the presence of $ \ctlcn{d}{_{A}}\neq 0 $ ---i.e., $ \Omega\phi_{1}\neq 0 $--- and $ \ct{C}{_{PQ}^{PQ}} $ ---i.e, $ \Omega\Re\prn{\phi_{2}}\neq0 $--- makes $ \ctncn{C}{_{A}}\neq 0 $. That is, the presence of `tangential' gravitational radiation and a Coulomb term sources	a `crossing' component. Only if these two last terms are absent, the news tensor $ \ctcn{N}{_{AC}} $ determines the `crossing' radiative components of the rescaled Weyl tensor at $ \N $.  \\
		
		Therefore, as in the case of a NEH, an adapted null hypersurface $ \N $ that is a null MTT also presents differences with respect to the radiative parts of the rescaled Weyl tensor at $ \scri $ in terms of the news tensor. Even though in this case the presence of crossing radiation components is possible, they do not depend only on $ \ctcn{N}{_{AB}} $. At $ \scri $, however, the terms containing $ \ctlcn{d}{_{A}} $ and $ \ct{C}{_{AB}^{AB}} $ vanish and $ \ctcn{N}{_{AB}} $ fully determines the outgoing gravitational radiation, as it is shown in full generality in \cref{sec:infinity}.
\section{Conclusions}
	Motivated by the radiative structure of $ \scri $ with $ \Lambda=0 $ and to study the properties of space-time null hypersurfaces $ \N $, a news tensor has been introduced (\cref{thm:news-N}) on a foliated null hypersurface geometrically, without using the field equations. Its expression in terms of kinematic quantities and the foliation function has been computed \eqref{eq:N-lied-sigma}. After that, assuming the CEFE with arbitrary $ \Lambda $,  an alternative expression in terms of the `asymptotic shear' has been derived \eqref{eq:news-asymptotic-shear}.  This news tensor, defined  on $ \N $, has the correct limit to $ \scri $ when $ \Lambda=0 $ ---in any conformal gauge--- and a generalised transport equation for the Geroch tensor has been given in \cref{eq:lied-rho-2,eq:lied-rho-1} at $ \scri $. However, the scenario in the  bulk has important differences with respect to $ \scri $, and the news tensor fails to determine, in general, the `crossing' radiative parts of the rescaled Weyl tensor on $ \N $. This should not be  surprising, as at $ \scri $ very special simplifications occur (see \cref{sec:infinity}); the most remarkable two are the well known umbilicity of $ \scri $ and the vanishing of the Weyl tensor $ \ct{C}{_{{\alpha\beta\gamma}}^{\delta}} $. It has to be remarked that these two properties are conformal-gauge-independent and that, further simplifications, such as endowing $ \scri $ with the structure of a NEH (or a MTT), are gauge-dependent assumptions. The information about \emph{gravitational radiation at $ \scri $ with $ \Lambda=0 $} is  contained in the rescaled Weyl tensor and \emph{fully determined by the news tensor}, which comes from the gauge-invariant part of the Schouten tensor. This is well known for divergence-free gauges, and here it has been proven in a covariant way for arbitrary conformal gauges.\\
	
	Finally, as a couple of proposals for future work, it would be interesting to study the news tensor on foliated null hypersurfaces in conformal space-time that `touch' $ \scri $  with arbitrary $ \Lambda $, and analyse its limit to infinity. In this way, it would be possible to relate this object with other asymptotic radiation conditions that have not been considered here, like the one presented in \cite{Fernandez-AlvarezSenovilla2020b}, and try to shed some light on the problem of the definition of the asymptotic gravitational energy in the case  $ \Lambda>0 $ \cite{Penrose2011,Szabados2019}. In particular, it would be possible then to explore the asymptotic behaviour of the news tensor on null horizons.

			\subsection*{Acknowledgments}

				 Work supported under Grant Margarita Salas MARSA22/20 (Spanish Ministry of Universities and European Union), financed by European Union -- Next Generation EU.

\printbibliography
\end{document}